\documentclass{tlp}

\usepackage[T1]{fontenc}

\newcommand{\ourtitle}{Towards Mass Spectrum Analysis with ASP}

\usepackage{amsmath}
\usepackage{amssymb}
\usepackage{mathtools}
\usepackage{nicematrix} \usepackage{marvosym} \usepackage{graphicx}
\usepackage{multirow}
\usepackage[pdftitle=\ourtitle,hidelinks]{hyperref}
\usepackage{subcaption}

\usepackage{xcolor}
\definecolor{strongteal}{HTML}{1f77b4}
\definecolor{strongorange}{HTML}{ff7f0e}
\definecolor{strongpurple}{HTML}{9467bd}
\definecolor{stronggreen}{HTML}{2ca02c}

\usepackage{listings}
\usepackage[frozencache,cachedir=_minted]{minted}
\IfFormatAtLeastTF{2024/01/01}{\newenvironment{ASPminted}[1][1]{\VerbatimEnvironment\begin{minted}[fontsize=\footnotesize, bgcolor=blue!5, escapeinside=||, linenos, firstnumber=#1, xleftmargin=10pt, numbersep=6pt]{asp-lexer.py:ASPLexer}}
    {\end{minted}}
\newcommand\ASPinline[1]{\mintinline[escapeinside=||]{asp-lexer.py:ASPLexer}{#1}}
}{\IfFormatAtLeastTF{2021/12/01}{\newenvironment{ASPminted}[1][1]{\VerbatimEnvironment\begin{minted}[fontsize=\footnotesize, bgcolor=blue!5, escapeinside=||, linenos, firstnumber=#1, xleftmargin=10pt, numbersep=6pt]{'asp-lexer.py:ASPLexer -x'}}
    {\end{minted}}
\newcommand\ASPinline[1]{\mintinline[escapeinside=||]{'asp-lexer.py:ASPLexer -f vs -x'}{#1}}
}{\newenvironment{ASPminted}[1][1]{\VerbatimEnvironment\begin{minted}[fontsize=\footnotesize, bgcolor=blue!5, escapeinside=||, linenos, firstnumber=#1, xleftmargin=10pt, numbersep=6pt]{asp-lexer.py:ASPLexer -x}}
    {\end{minted}}
\newcommand\ASPinline[1]{\mintinline[escapeinside=||]{asp-lexer.py:ASPLexer -f vs -x}{#1}}
}
}

\usepackage{tikz}
\usetikzlibrary{graphs,positioning,calc}

\usepackage{multicol}

\usepackage{chemfig}
\newlength{\mcfvspace}
\def\emails#1#2#3{{\normalfont\rmfamily
  \itshape\textup{(}e-mails: \textup{\texttt{\href{mailto:#1}{\color{blue}#1},}}\ \textup{\texttt{\href{mailto:#2}{\color{blue}#2},}} \textup{\texttt{\href{mailto:#3}{\color{blue}#3}}})}}

\usepackage{todonotes}

\newcommand{\opfont}[1]{\text{\sf{#1}}} \newcommand{\ltuple}{\langle}
\newcommand{\rtuple}{\rangle}
\newcommand{\tuple}[1]{\ltuple{#1}\rtuple}

\usepackage{thmtools}

 \declaretheorem[sharenumber=dummytheorem]{proposition}
\declaretheorem[sharenumber=dummytheorem]{example}
\declaretheorem[sharenumber=dummytheorem]{definition}

\newenvironment{talign}
  {\align}
  {\endalign}
\newenvironment{talign*}
  {\csname align*\endcsname}
  {\endalign}
\newenvironment{tequation}
  {\equation}
  {\endequation}
\newenvironment{tequation*}
  {\csname equation*\endcsname}
  {\endequation}

\begin{document}
\title{\ourtitle}
\lefttitle{Nils Küchenmeister, Alex Ivliev, and Markus Krötzsch}
\jnlPage{1}{22}
\jnlDoiYr{2025}
\doival{10.1017/xxxxx}
\pdfimageresolution=192 \begin{authgrp}
\author{\gn{Nils Küchenmeister}\orcid[0009-0004-0376-0328], \gn{Alex Ivliev}\orcid[0000-0002-1604-6308], and \gn{Markus Krötzsch}\orcid[0000-0002-9172-2601]}
\affiliation{
  TU Dresden, Germany\\
  \emails{nils.kuechenmeister@tu-dresden.de}{alex.ivliev@tu-dresden.de}{markus.kroetzsch@tu-dresden.de}
}
\end{authgrp}  
\maketitle \begin{abstract}
We present a new use of Answer Set Programming (ASP) to discover the molecular structure
of chemical samples based on the relative abundance of elements and structural fragments,
as measured in mass spectrometry. 
To constrain the exponential search space for this combinatorial problem, we develop
canonical representations of molecular structures and an ASP implementation that
uses these definitions. We evaluate the correctness of our implementation over a large
set of known molecular structures, and we compare its quality and performance to other ASP
symmetry-breaking methods and to a commercial tool from analytical chemistry.
Under consideration in Theory and Practice of Logic Programming (TPLP).
 \end{abstract}

\section{Introduction}

Mass spectrometry is a powerful technique to determine the chemical composition of a substance \cite[]{massspec}.
However, the mass spectrum of a substance does not reveal its exact molecular structure,
but merely the possible ratios of elements in the compound and its fragments.
To identify a sample, researcher may use commercial databases 
(for common compounds), or software tools that can discover 
molecular structures from the partial information available.
The latter leads to a combinatorial search problem that is a natural fit 
for answer set programming (ASP).
Molecules can be modeled as undirected graphs, representing the different elements
and atomic bonds as node and edge labels, respectively.
ASP is well-suited to encode chemical domain knowledge 
(e.g., possible number of bonds for carbon) and extra information about the sample
(e.g., that it has an $\textit{OH}$ group), so that each answer set encodes a 
valid molecular graph.

Unfortunately, this does not work: a direct ASP encoding
yields exponentially many answer sets for each molecular graph due to the large number of
symmetries (automorphisms) in such graphs.
For example, $\textit{C}_6\textit{H}_{12}\textit{O}$ admits 211 distinct molecule structures but leads to 111,870
answer sets. 
Removing redundant solutions and limiting the search to unique representations
are common techniques used in the ASP community 
where they have motivated research on \emph{symmetry-breaking}.
Related approaches work by rewriting the ground program before solving, see \cite[\textsc{sbass}]{sbass} and \cite[\textsc{BreakID}]{breakid_sat,breakid_system_description},
or augmenting ASP programs with additional constraints learned from generated instances \cite[\textsc{ilasp} using \textsc{sbass}]{learning}.
Some methods also integrate symmetry-breaking with existing solvers \cite[\textsc{idp3}]{localdomain},
or provide dedicated solvers \cite[\textsc{HC-asp}]{newsearch}.
In addition to these general approaches, 
there are also methods that explicitly 
define symmetry-breaking constraints for undirected graphs \cite[]{graph}.
However, our experiments with some of these approaches
still produced 10--10,000 times more answer sets than molecules
even in simple cases.

We therefore develop a new approach that prevents symmetries in graph representations already during grounding,
and use it as the core of an ASP-based prototype implementation for enumerating molecular structures based on 
partial chemical information.
In Section~\ref{sec:tool}, we explain the problem and our prototype tool from a user perspective.
We then define the problem formally in Section~\ref{sec_prob}, using an abstract notion of 
\emph{tree representations} of molecular graphs that takes inspiration from the chemical notation SMILES \cite[]{smiles}.
We then derive a new canonical representation for molecular graphs (Section~\ref{sec-encoding}) to
guide our ASP implementation (Section~\ref{sec-impl}).
In Section~\ref{sec_experiments}, we evaluate the correctness, symmetry-breaking capabilities, and performance
of our tool in comparison to other ASP-based approaches and a leading commercial software for analytical chemistry \cite[]{gugisch2015molgen}.
We achieve perfect symmetry-breaking, i.e., the removal of all redundant solutions,
for acyclic graph structures and up to three orders of magnitude reduction in 
answer sets for cyclic cases in comparison to other ASP approaches. Overall, ASP therefore appears to be a promising basis for
this use case, and possibly for other use cases concerned with undirected graph structures.

Our ASP source code, evaluation helpers, and data sets are available online at
\url{https://github.com/knowsys/eval-2024-asp-molecules}. The sources of our prototype application are at
\url{https://gitlab.com/nkuechen/genmol/}. 
This work is an extended version of the conference paper \cite[]{KIK2024},
which was presented at LPNMR 2024.

 \section{Analysis of Mass Spectra with \textsc{Genmol}}\label{sec:tool}

Many mass spectrometers break up samples into smaller fragments and
measure their relative abundance. 
The resulting mass spectrum forms a characteristic pattern, 
enabling inferences about the underlying sample.
High-resolution spectra may contain information such as 
``the molecule has six carbon atoms''
or ``there is an $\textit{OH}$ group'',
but cannot reveal the samples's full molecular structure. 
In chemical analysis, we are looking for 
molecular structures that are consistent with the measured mass spectrum.

\begin{figure}[t!]
  \centering
  \includegraphics[width=\textwidth]{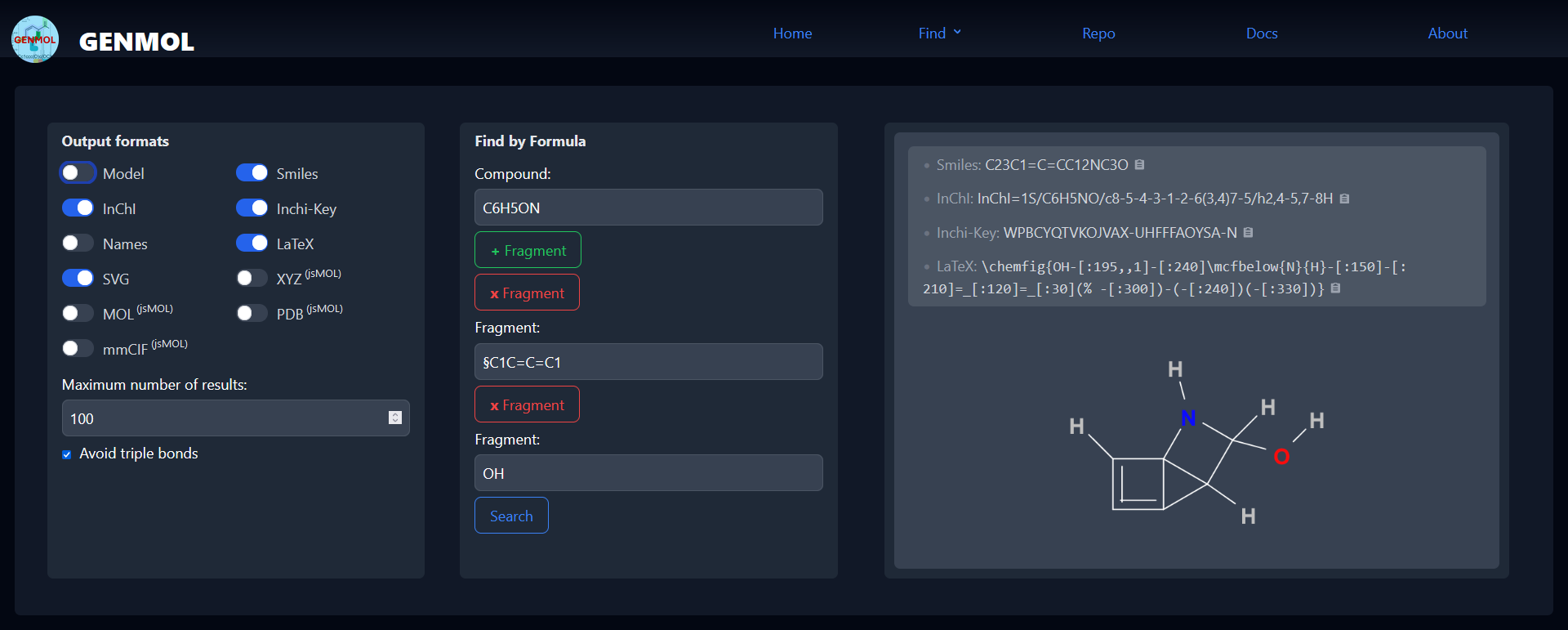}
  \caption{User interface of \textsc{Genmol}}
  \label{fig:genmol}
\end{figure}
To address this task, we have developed \textsc{Genmol},
a prototype application for enumerating molecular structures for a given composition of fragments.
It is available as a command-line tool and as a progressive web application (PWA), shown in Fig.~\ref{fig:genmol}.
\textsc{Genmol} is implemented in Rust, with the web front-end using the Yew framework on top of a JSON-API,
whereas the search for molecular structures is implemented in Answer Set Programming (ASP) and solved using
\emph{clingo} \cite[]{clingo}.
An online demo of \textsc{Genmol} is available for review at \url{https://tools.iccl.inf.tu-dresden.de/genmol/}.

The screenshot shows the use of \textsc{Genmol} with a \emph{sum formula} $C_6H_5ON$ and two fragments as input.
Specifying detected fragments and restricting bond types helps to reduce the search space.
Alternatively, users can provide a molecular mass or a complete mass spectrum,
which will then be associated with possible chemical formulas using, e.g., 
information about the abundance of isotopes.

The core task of \textsc{Genmol} then is to find molecules that match the given input constraints.
Molecules in this context are viewed as undirected graphs of atoms, 
linked by covalent bonds
that result from sharing electrons.\footnote{This graph does not always determine the spacial configuration of molecules, which cannot be determined by mass spectrometry alone, yet it suffices for many applications.}
Many chemical elements admit a fixed number of bonds, the so-called \emph{valence},
according to the number of electrons available for binding
(e.g., carbon has a valence of 4). Bonds may involve several electrons, leading to single, double,
triple bonds, etc. The graph structure of molecules, the assignment of elements, and the possible
types of bonds can lead to a large number of possible molecules for a single chemical formula,
and this combinatorial search task is a natural match for Answer Set Programming.

 \section{Problem Definition: Enumeration of Molecules}\label{sec_prob}

We begin by formalizing the problem of molecule enumeration, and 
by introducing a chemistry-inspired representation of molecules.
We consider a set $\mathbb{E}$ of \emph{elements}, with each $e\in\mathbb{E}$ associated
with a \emph{valence} $\mathbb{V}(e)\in\mathbb{N}_{>0}$.
We assume that $\mathbb{E}$ contains a distinguished element $H\in\mathbb{E}$
(hydrogen) with $\mathbb{V}(H)=1$.
Note that in reality, elements may have multiple valences.
To simplify presentation, we will ignore this aspect in the formalization.
The implementation, however, supports multi-valence elements 
by representing them as separate elements. 
We model molecules as undirected graphs with edges labelled by natural numbers to
indicate the type of bond.

\begin{definition}\label{def_molgraph}
A \emph{molecular graph} $G$ is a tuple $G=\tuple{V, E, \ell, b}$
with vertices $V=\{1,\ldots,k\}$ for some $k\geq 1$,
undirected edges $E \subseteq \binom{V}{2}$, where $\binom{V}{2}$ is the set of all 2-element subsets of $V$,
and labelling functions $\ell \colon V \to \mathbb{E}$ and $b \colon E \to \mathbb{N}_{>0}$.

The \emph{degree} $\opfont{deg}(v)$ of a vertex $v\in V$ is defined as $\opfont{deg}(v)=\sum\{b(e)\mid e\in E, v\in e\}$.
A list of $n$ distinct vertices $v_1, \ldots, v_n$ is a \emph{simple path} in $G$ if $\{v_i, v_{i+1}\} \in E$ for every $1 \leq i < n$.
The graph $G$ is connected if there is a simple path from $v$ to $w$ for every pair $v, w \in V$.
\end{definition}

\noindent Since we assume that atoms in a molecule use all available bonds according to their valence,
the (very frequent) hydrogen atoms do not need to be mentioned explicitly. 
Such \emph{hydrogen-suppressed} molecular graphs are common in computational chemistry,
and the basis for the next definition:

\begin{definition}\label{def:validgraph}
A \emph{molecular formula} is a function $f\colon\mathbb{E}\to\mathbb{N}$.
A (hydrogen-suppressed) molecular graph $G=\tuple{V, E, \ell, b}$ is \emph{valid} for $f$,
if it satisfies the following properties:
\begin{enumerate}[label=(\theenumi),leftmargin=25pt]
  \item $G$ is connected,\label{it_validgraph_connect}
  \item for every $e \in \mathbb{E}$ with $e \neq H$, $\#\{v \in V \mid \ell(v) = e \} = f(e)$,\label{it_validgraph_elemsum}
  \item for every $v \in V$, $\opfont{deg}(v) \leq \mathbb{V}(\ell(v))$,\label{it_validgraph_valence}
  \item $\sum_{v \in V} \left(\mathbb{V}(\ell(v)) - \opfont{deg}(v)\right) = f(H)$.\label{it_validgraph_hydrogen}
\end{enumerate}
\end{definition}

\noindent The enumeration problem can now be stated as follows:
for a given molecular formula $f$, enumerate, up to isomorphism, all valid molecular graphs for $f$.
In general, the number of distinct isomorphic molecular graphs is exponential in the number of atoms.
It is therefore important to reduce the enumeration of redundant isomorphic graphs.

A first step towards this is the use of a more restricted representation of molecular graphs.
Here, we take inspiration from the \emph{simplified molecular-input line-entry system} (\emph{SMILES}),
a widely used serialization format for molecular graphs.
SMILES strings start from an (arbitrary) spanning tree of the molecular graph,
serialized in a depth-first order, with subtrees enclosed in parentheses.
Edges not covered by the spanning tree (since they would complete a cycle) are indicated
by pairs of numeric \emph{cycle markers}.

\newcommand{\vnum}[1]{^{\textcolor{stronggreen}{#1}}}
\newcommand{\natom}[2]{\mathit{#1}\vnum{#2}}
\begin{figure}[tbp]
  \centering
  \begin{minipage}{0.9\textwidth}
    \begin{multicols*}{2}
      \begin{center}
      \scalebox{1.0}{\chemfig{\mbox{}N\vnum{6}-[:300,,2,,strongteal]\natom{C}{5}=[:240,,,,strongteal]\natom{N}{4}-[:300,,,,strongteal]\natom{C}{3}=[,,,,strongteal]\natom{N}{2}-[:60,,,,strongteal]\natom{C}{1}=[:120,,,,strongteal]\natom{C}{7}(-[:180,,,,strongorange]\phantom{\natom{C}{8}})-[:48,,,,strongteal]\natom{N}{8}=[:336,,,1,strongteal]\natom{C}{9}-[:264,,1,1,strongteal]\natom{N}{10}(-[:192,,1,,strongorange]\phantom{\natom{C}{8}})}}
      \end{center}
\columnbreak
\begin{center}
      \scalebox{0.8}{
        \begin{tikzpicture}
          \tikzset{>=latex}
          \graph[grow down=0.65cm, branch right=0.45cm, nodes={circle,inner sep=0,minimum size=0cm}]{
            c1/$\natom{C}{1}$ [xshift=2.0cm] --[strongteal] {
              n2/$\natom{N}{2}$ [xshift=1.6cm] --[double, strongteal] c3/$\natom{C}{3}$ [xshift=1.2cm] --[strongteal] n4/$\natom{N}{4}$ [xshift=0.8cm] --[double, strongteal] c5/$\natom{C}{5}$ [xshift=0.4cm] --[strongteal] n6/$\natom{N}{6}$,
              c7/$\natom{C}{7}$ [xshift=2.0cm, > double] --[strongteal] n8/$\natom{N}{8}$ [xshift=2.4cm] --[double, strongteal] c9/$\natom{C}{9}$ [xshift=2.8cm] --[strongteal] n10/$\natom{N}{10}$ [xshift=3.2cm]
            };
c5/$\natom{C}{5}$ --[thick,bend right, dotted, strongorange,double distance=0pt] c7/$\natom{C}{7}$;
            c1/$\natom{C}{1}$ --[thick,bend left, dotted, strongorange,double distance=0pt] n10/$\natom{N}{10}$;
          };
        \end{tikzpicture}
      }
   \end{center}
    \end{multicols*}
  \end{minipage}
  \caption{Hydrogen-suppressed molecular graph of adenine ($\mathit{C}_5\mathit{H}_5\mathit{N}_5$) and corresponding spanning tree with cycle edges (dotted); superscripts indicate correspondence of vertices}\label{fig-example-adenine}
\end{figure}
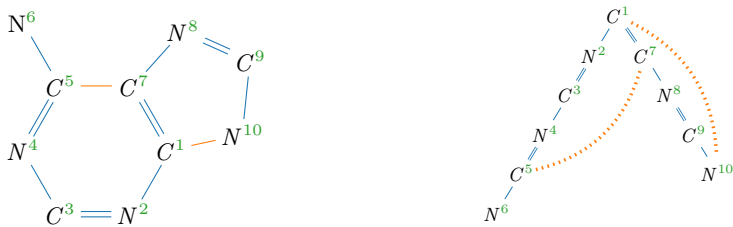
\begin{example}\label{example:problem}
Adenine ($\textit{C}_5\textit{H}_5\textit{N}_5$) has the
graph structure shown in Figure~\ref{fig-example-adenine} (left), with a spanning
tree on the right.
Hydrogen atoms are omitted, as they take up any free binding place.
For instance, $\natom{C}{3}$ has a double- and a single-bond and a valence of four, leaving one hydrogen atom.
The SMILES is \texttt{C1(N=CN=C2N)=C2N=CN1}
where consecutive atoms are connected by bonds and double bonds are marked (=). 
The segment from vertex 2 to 6 is in parentheses to indicate a branch.
Additionally, 
the two non-sequential (dotted) connections 
are denoted by matching pairs of numerical markers.
\end{example}

\vspace*{1em}%
\begin{definition}\label{def:tree-representation}
A molecular graph $G=\tuple{V, E, \ell, b}$ is a \emph{molecular tree} if
it is free of cycles and the natural order of vertices $V=\{1,\ldots,k\}$
corresponds to a depth-first search of the tree (in particular, vertex $1$ is
the root).

A \emph{tree representation} of an arbitrary molecular graph $G=\tuple{V, E, \ell, b}$
is a set $T\subseteq E$ such that $\tuple{V, T, \ell, b}$ is a molecular tree.
In this case, we denote $G$ as $\tuple{V, T \cup C, \ell, b}$ where
$T$ are the \emph{tree edges} and $C=E\setminus T$ are the \emph{cycle edges}.
\end{definition}

\noindent Note that the tree edges $T$ by definition visit every vertex of $V$, and 
the cycle edges $C$ merely make additional connections between vertices of the tree.
In SMILES, the edges in $C$ and their labels (bond types) are encoded with special markers,
while the order of vertices is given by the order of appearance in the SMILES string.

A tree representation for a given molecular graph is uniquely determined by
the following choices:
(a) a spanning tree (corresponding to a choice of tree edges and cycle edges),
(b) a root vertex, and
(c) for every vertex, an order of visiting its child vertices in depth first search.
For a given graph structure, the number of tree representations can be significantly
lower than the number of isomorphic molecular graphs. For example, a graph that is a chain
has only linearly many tree representations, but still admits exponentially many graphs.
Nevertheless, the number of tree representations can still be exponential, and we will
investigate below how the choices in (a)--(c) can be further constrained
to reduce redundancy.

 \section{Canonical tree representations of molecular graphs}\label{sec-encoding}

To eliminate redundant isomorphic solutions, we first define a canonical tree representation of
any molecular graph. The defining conditions of this unique representation will then be used to 
constrain the search for possible graphs in our implementation. We first consider the simpler case
of molecular trees.

\subsection{Canonical Molecular Trees}\label{sec_canon_tree}

We define a total order on molecular trees, 
which will allow us to define a largest tree among a
set of candidates.
To define this order inductively, 
we need to consider subtrees that may not have $1$ as their root.
For a molecular tree $G=\tuple{V, E, \ell, b}$ 
with vertex $v\in V$, let $\opfont{subtree}(G,v)$
be the tuple $\tuple{V', E', \ell, b, \opfont{in}(G,v)}$ 
where $V'$ and $E'$ are the restriction of $V$ and $E$, respectively, 
to vertices that are part of the subtree with root $v$ in $G$,
and either $\opfont{in}(G,v)=0$ if $v=1$ is the root of $G$, or $\opfont{in}(G,v)=b(e)$ is the edge label $b(e)\geq 1$
of the edge $e$ between $v$ and its parent in $G$.
Moreover, if $\tuple{c_1,\ldots,c_k}$ are the ordered children of $v$, then 
$\opfont{childtrees}(G,v)=\tuple{\opfont{subtree}(G,c_1),\ldots,\opfont{subtree}(G,c_k)}$.

Finally, let $R(G,v)=\tuple{d,s,c,\ell,b}$ be the tuple with
$d\geq 1$ the depth of $\opfont{subtree}(G,v)$;
$s\geq 1$ the size (number of vertices) of $\opfont{subtree}(G,v)$;
$c\geq 0$ the number of children of $v$;
$\ell\in\mathbb{E}$ the element $\ell(v)$ of $v$; and
$b=\opfont{in}(G,v)\geq 0$.
Intuitively, we use $R(G, v)$ to define an order between nodes,
which we extend to an order between molecular trees.
For the following definition, recall that the \emph{lexicographic extension} of
a strict order $\prec$ to tuples of the same size $k$ is defined by setting
$\vec{t}\prec\vec{u}$ if there is $i\in\{1,\ldots,k\}$ such that
$\vec{t}[i]\prec\vec{u}[i]$ and $\vec{t}[j]=\vec{u}[j]$ for all $j<i$.

\vspace*{1em}%
\begin{definition}\label{def_tree_order}
Let $\sqsubset$ be an arbitrary but fixed strict total order on $\mathbb{E}$,
overloaded to also denote the usual order $<$ on natural numbers.
We further extend $\sqsubset$ to 5-tuples of the form $R(G,v)$ lexicographically.

We define a strict order $\prec$ on subtrees as the smallest relation 
where, for each pair of subtrees $S_i=\opfont{subtree}(G_i,v_i)$ of molecular trees $G_i$ ($i=1,2$),
$S_1\prec S_2$ holds if
\begin{enumerate}[label=(\theenumi),leftmargin=25pt]
\item $R(G_1,v_1)\sqsubset R(G_2,v_2)$, or \label{item_tree_order_tuple}
\item $R(G_1,v_1)=R(G_2,v_2)$, i.e., $S_1$ and $S_2$ are locally indistinguishable, with (necessarily equal) number of children $k$, such that
    $\opfont{childtrees}(G_1,v_1)\prec\opfont{childtrees}(G_2,v_2)$ where $\prec$ is the lexicographic extension of $\prec$ to $k$-tuples of subtrees.
    \label{item_tree_order_children}
\end{enumerate}
For molecular trees $G_1$ and $G_2$, we define $G_1\prec G_2$ if $\opfont{subtree}(G_1,1)\prec\opfont{subtree}(G_2,1)$.
\end{definition}

\begin{example}
A tree representation of threonine,
annotated with the respective $R(G, v)$ tuple
for all of its nodes, is shown in Figure~\ref{fig-example-threonine}.
In two instances, pairs of nodes share the same $R(G, v)$ value.
For example, 
nodes $\natom{C}{2}$ and $\natom{C}{5}$ are indistinguishable 
by Definition~\ref{def_tree_order}~\ref{item_tree_order_tuple},
but Definition~\ref{def_tree_order}~\ref{item_tree_order_children} 
yields $\opfont{subtree}(G,2) \prec \opfont{subtree}(G,5)$.

\end{example}
    
\begin{figure}[tbp]
\centering
\begin{minipage}{0.5\textwidth}
    \begin{flushright}
        $R(G,1)=       \tuple{3,8,3,C,0}$\\
        $R(G,2)=R(G,5)=\tuple{2,3,2,C,1}$\\
        $R(G,3)=       \tuple{1,1,0,O,2}$\\
        $R(G,4)=R(G,6)=\tuple{1,1,0,O,1}$\\
        $R(G,7)=       \tuple{1,1,0,C,1}$\\
        $R(G,8)=       \tuple{1,1,0,N,1}$
    \end{flushright}
\end{minipage}\begin{minipage}{0.5\textwidth}
    \begin{center}
    \scalebox{0.8}{
        \begin{tikzpicture}
        \tikzset{>=latex}
        \graph[grow down=0.65cm, branch right=0.45cm, nodes={circle,inner sep=0,minimum size=0cm}]{
            c1/$\natom{C}{1}$ [xshift=0.8cm, draw=red] --[strongteal] {
                c2/$\natom{C}{2}$ [xshift=0.4cm] --[strongteal] {
                    o3/$\natom{O}{3}$ [> double],
                    o4/$\natom{O}{4}$ [xshift=0.4cm]
                },
                c5/$\natom{C}{5}$ [xshift=0.4cm]--[strongteal] {
                    o6/$\natom{O}{6}$ [xshift=0.8cm],
                    c7/$\natom{C}{7}$ [xshift=1.2cm]
                },
                n8/$\natom{N}{8}$ [xshift=1.2cm]
            };
        };
        \end{tikzpicture}
    }
    \end{center}
\end{minipage}
\caption{Canonical molecular tree of threonine ($\mathit{C}_4\mathit{H}_9\mathit{N}\mathit{O}_3$); central vertex is circled}\label{fig-example-threonine}
\end{figure}
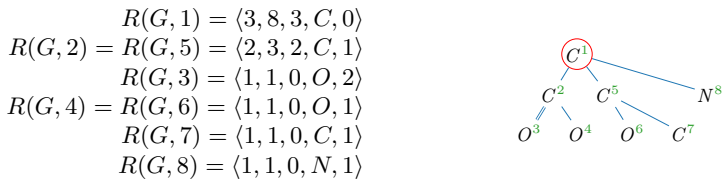

\begin{proposition}\label{prop_prec_total}
The relation $\prec$ of Definition~\ref{def_tree_order} is a strict total order on molecular trees.
\end{proposition}
\begin{proof}
The claim follows by showing that $\prec$ is a strict order on subtrees.
The order $\sqsubset$ on tuples is a strict total order 
since it is the lexicographic extension of strict total orders.
Hence, all subtrees $S_i=\opfont{subtree}(G_i,v_i)$ ($i=1,2$) with $R(G_1,v_1)\neq R(G_2,v_2)$ are $\prec$-comparable by \ref{item_tree_order_tuple}.
Totality for case $R(G_1,v_1)= R(G_2,v_2)$ is shown by induction on the (equal) depth of the subtrees $S_1$ and $S_2$.
For depth $1$, $S_1$ and $S_2$ have no children, and $R(G_1,v_1)=R(G_2,v_1)$ implies $S_1=S_2$.
For depth $i$ with $i>1$, we can assume all subtrees of depth $\leq i-1$ to be $\prec$-comparable unless equal. 
If $\opfont{childtrees}(G_1,v_1)$ and $\opfont{childtrees}(G_2,v_2)$ are comparable under the lexicographic extension of $\prec$,
then $S_1$ and $S_2$ are $\prec$-comparable by \ref{item_tree_order_children}.
Otherwise, $\opfont{childtrees}(G_1,v_1)=\opfont{childtrees}(G_2,v_2)$, and therefore $S_1=S_2$.
\end{proof}

\noindent By Proposition~\ref{prop_prec_total}, we could define the canonical molecular tree to be the $\prec$-largest tree among
a set of isomorphic trees. However, this would force us to select a root that is the start (or end) of a longest path
in the graph to maximize the depth of the tree. 
In general, many such nodes might exist.
It is therefore more efficient to compare a smaller
set of potential roots that are closer together:

\begin{definition}\label{def_canon_tree}
Let $G=\tuple{V, E, \ell, b}$ be a molecular tree.
A vertex $v_i$ is \emph{central} in a simple path $v_1,\ldots,v_n$ in $G$ if $i\in\{\lceil(n+1)/2\rceil,\lfloor(n+1)/2\rfloor\}$
(a singleton set if $n$ is odd). A vertex is central in $G$ if it is central in
any longest simple path in $G$.

The \emph{canonical molecular tree} $C$ of $G$ is
the $\prec$-largest molecular tree that is obtained by permutation of vertices in $G$
such that the root of $C$ is central in $G$.
\end{definition}

\noindent In every tree, the central vertices of all longest simple paths are the same,
and hence there are at most two.
Indeed, two distinct longest paths always share at least one vertex in a tree. So if two such paths
$\vec{v}$ and $\vec{w}$ would have different central vertices $v_a\neq w_b$,
and a shared vertex $v_i=w_j$ with (w.l.o.g.) $a>i$ and $b>j$, then the path
$v_1,\ldots,v_a,\ldots,v_i=w_j,\ldots,w_b,\ldots,w_1$ would be longer than $\vec{v}$ and $\vec{w}$, contradicting their 
assumed maximal length.
Using this insight, our implementation can find the canonical molecular tree by considering at most two possible roots.

\begin{example}
Two alternative tree representations for glycine are depicted in Figure~\ref{fig-example-glycine}.
The red circles indicate central vertices, 
which are considered as candidates for the root of the spanning tree.
Comparing the roots of both trees,
we find that the tree representation on the right
has more children, while the corresponding subtrees are identical in 
depth and number of vertices.
Therefore, $R(G,1) \sqsubset R(G',1)$,
which implies that the right-hand tree is {$\prec$-larger}
and thus the canonical molecular tree by Definition~\ref{def_canon_tree}.
\end{example}

\begin{figure}[tbp]
\centering
\begin{minipage}{0.29\textwidth}
    \begin{flushright}
        $R(G,1)=\tuple{3,5,2,C,0}$\\
        $R(G,2)=\tuple{2,3,2,C,1}$\\
        $R(G,3)=\tuple{1,1,0,O,2}$\\
        $R(G,4)=\tuple{1,1,0,O,1}$\\
        $R(G,5)=\tuple{1,1,0,N,1}$
    \end{flushright}
\end{minipage}\begin{minipage}{0.2\textwidth}
    \begin{center}
    \scalebox{0.8}{
        \begin{tikzpicture}
        \tikzset{>=latex}
        \graph[grow down=0.65cm, branch right=0.45cm, nodes={circle,inner sep=0,minimum size=0cm}]{
            c1/$\natom{C}{1}$ [xshift=0.8cm, draw=red] --[strongteal] {
                c2/$\natom{C}{2}$ [xshift=0.4cm, draw=red] --[strongteal] {
                    o3/$\natom{O}{3}$ [> double],
                    o4/$\natom{O}{4}$ [xshift=0.4cm]
                },
                n5/$\natom{N}{5}$ [xshift=0.4cm]
            };
        };
        \end{tikzpicture}
    }
    \end{center}
\end{minipage}\begin{minipage}{0.01\textwidth}
    \begin{center}
        $\prec$
    \end{center}
\end{minipage}\begin{minipage}{0.2\textwidth}
    \begin{center}
    \scalebox{0.8}{
        \begin{tikzpicture}
        \tikzset{>=latex}
        \graph[grow down=0.65cm, branch right=0.45cm, nodes={circle,inner sep=0,minimum size=0cm}]{
            c1/$\natom{C}{1}$ [xshift=0.8cm, draw=red] --[strongteal] {
                c2/$\natom{C}{2}$ [xshift=0.4cm, draw=red] --[strongteal] n3/$\natom{N}{3}$,
                o4/$\natom{O}{4}$ [xshift=0.8cm, > double],
                o5/$\natom{O}{5}$ [xshift=1.2cm]
            };
        };
        \end{tikzpicture}
    }
    \end{center}
\end{minipage}\begin{minipage}{0.29\textwidth}
    $R(G',1)=\tuple{3,5,3,C,0}$\\
    $R(G',2)=\tuple{2,2,1,C,1}$\\
    $R(G',3)=\tuple{1,1,0,N,1}$\\
    $R(G',4)=\tuple{1,1,0,O,2}$\\
    $R(G',5)=\tuple{1,1,0,O,1}$
\end{minipage}
\caption{Molecular trees of glycine ($\mathit{C}_2\mathit{H}_5\mathit{N}\mathit{O}_2$); central vertices are circled}\label{fig-example-glycine}
\end{figure}
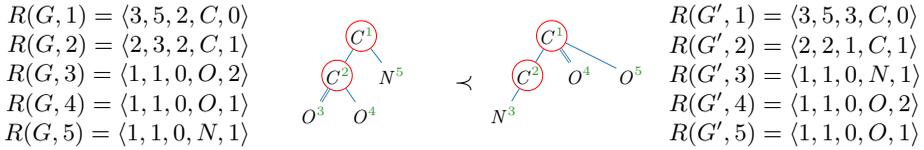

\subsection{Canonical Molecular Graphs}\label{sec_canon_graph}

Next, we define a canonical tree representation for arbitrary molecular graphs,
allowing us to extend the above order on trees to an order on graphs.
Let $G=\tuple{V, T \cup C, \ell, b}$ be a tree representation with $V=\{1,\ldots,k\}$.
We construct a molecular tree $G'=\opfont{tr}(G)$ by replacing each cycle 
edge $\{v, w\}$ in $C$ with two new edges, one from $v$ and one from $w$
leading to fresh vertices.
Hence, let $V' = \{k + 1, \dots, k + 2 \cdot |C|\}$
be the set of fresh vertices
and assume that the cycle edges $C = \{c_1, \dots, c_m\}$
are numbered in an arbitrary way.
The edges in $C$ are replaced by new tree edges
\begin{equation*}
    T' = \{ \{\min c_i, k + 2i - 1\}, \{\max c_i, k + 2i\} \mid c_i \in C \}.
\end{equation*}
Each new edge inherits the label of its corresponding cycle edge,
and each new vertex receives the label of the node 
that was connected to its parent in the cycle edge.
Formally, we obtain the labeling functions 
$b'(\{v, n\}) = b(c_i)$ and $\ell'(n) = \ell(w)$
for $c_i = \{v, w\} \in C$, $\{v, n\} \in T'$, 
with $v, w \in V$ and $n \in \{k + 2i - 1, k + 2i\} \subseteq V'$.
Thus, the transformed molecular tree is 
$G' = \tuple{V \cup V', T \cup T', \ell \cup \ell', b \cup b'}$.

\vspace*{1em}%
\begin{figure}[tbp]
  \centering
  \begin{minipage}{0.9\textwidth}
    \begin{multicols*}{2}
      \begin{center}
      \scalebox{0.8}{
        \begin{tikzpicture}
          \tikzset{>=latex}
          \graph[grow down=0.65cm, branch right=0.45cm, nodes={circle,inner sep=0,minimum size=0cm}]{
            c1/$\natom{C}{1}$ [xshift=2.0cm] --[strongteal] {
              n2/$\natom{N}{2}$ [xshift=1.6cm] --[double, strongteal] c3/$\natom{C}{3}$ [xshift=1.2cm] --[strongteal] n4/$\natom{N}{4}$ [xshift=0.8cm] --[double, strongteal] c5/$\natom{C}{5}$ [xshift=0.4cm] --[strongteal] n6/$\natom{N}{6}$,
              c7/$\natom{C}{7}$ [xshift=2.0cm, > double] --[strongteal] n8/$\natom{N}{8}$ [xshift=2.4cm] --[double, strongteal] c9/$\natom{C}{9}$ [xshift=2.8cm] --[strongteal] n10/$\natom{N}{10}$ [xshift=3.2cm]
            };
c5/$\natom{C}{5}$ --[thick,bend right, dotted, strongorange,double distance=0pt] c7/$\natom{C}{7}$;
            c1/$\natom{C}{1}$ --[thick,bend left, dotted, strongorange,double distance=0pt] n10/$\natom{N}{10}$;
          };
        \end{tikzpicture}
      }
      \end{center}
      \columnbreak
      \begin{center}
      \scalebox{0.8}{
        \begin{tikzpicture}
          \tikzset{>=latex}
          \graph[grow down=0.65cm, branch right=0.45cm, nodes={circle,inner sep=0,minimum size=0cm}]{
            c1/$\natom{C}{1}$ [xshift=2.0cm] --[strongteal] {
              n2/$\natom{N}{2}$ [xshift=1.6cm] --[double, strongteal] c3/$\natom{C}{3}$ [xshift=1.2cm] --[strongteal] n4/$\natom{N}{4}$ [xshift=0.8cm] --[double, strongteal] c5/$\natom{C}{5}$ [xshift=0.4cm] --[strongteal] {
                n6/$\natom{N}{6}$,
                tr_c7/$\bar{\natom{C}{7}}$ [xshift=.4cm,semitransparent],
              },
              c7/$\natom{C}{7}$ [xshift=1.6cm, > double] --[strongteal] {
                tr_c5/$\bar{\natom{C}{5}}$ [xshift=1.2cm,semitransparent],
                n8/$\natom{N}{8}$ [xshift=1.6cm] --[double, strongteal] c9/$\natom{C}{9}$ [xshift=2.0cm] --[strongteal] n10/$\natom{N}{10}$ [xshift=2.4cm] --[strongteal] tr_c1/$\bar{\natom{C}{1}}$ [xshift=2.8cm,semitransparent]
              },
              tr_n10/$\bar{\natom{N}{10}}$ [xshift=1.8cm,semitransparent],
            };
};
        \end{tikzpicture}
      }
      \end{center}
    \end{multicols*}
  \end{minipage}
  \caption{Tree representation of adenine and molecular tree with replaced cycle edges}\label{fig-example-adenine-VC}
\end{figure}
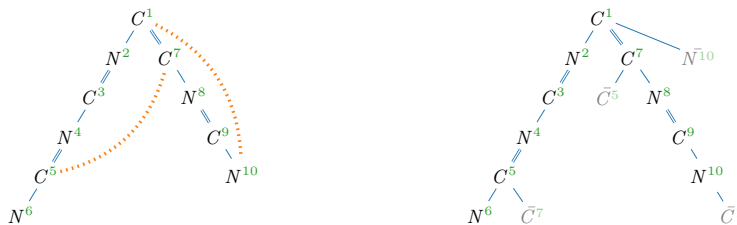
\begin{example}[Continuation of Example~\ref{example:problem}]
    Figure~\ref{fig-example-adenine-VC} shows the tree representation $G$ of Adenine on the left,
    which has two cycle edges (depicted in orange).
    Modifying it by adding two fresh nodes each, the molecular tree $\opfont{tr}(G)$ shown on the right emerges.
\end{example}

Given two tree representations $G_1$ and $G_2$, we define $G_1\prec G_2$ if
$\opfont{tr}(G_1)\prec \opfont{tr}(G_2)$. This does not define a total order, since $\opfont{tr}$ is not
injective. However, on any set of tree representations with the same number of vertices (and especially
on any set of isomorphic tree representations), $\opfont{tr}$ is injective and $\prec$ is total.

Though $\prec$ defines a largest tree representation of any molecular graph, it is impractical to
consider every possible such representation in search of this optimum. We therefore restrict
to tree representations where the tree edges are identified by iterative addition of longest simple paths
that do not create cycles.

\begin{definition}\label{def_refinement}
A \emph{pre-tree representation} is a molecular graph $G=\tuple{V, E, \ell, b}$ where $E$ is a disjoint union $E=T\cup C$
such that the edges of $T$ define a tree (possibly not a spanning tree for $G$).

An \emph{extension} of $G$ is a simple path $v_1,\ldots,v_n$ 
such that $v_1\in T$ and $v_2,\ldots,v_n\in C$
or $v_1,\ldots,v_n \in C$ if $T = \emptyset$.
A \emph{longest extension} is one of maximal length among all extensions of $G$.
A \emph{refinement} of $G$ is a pre-tree representation $G'=\tuple{V, T'\cup C', \ell, b}$, where
$T'=T\cup P$ and $C'=C\setminus P$ for a set of edges $P$ of some longest extension of $G$.
\end{definition}

We can view any molecular graph as a \emph{pre-tree representation} with $T=\emptyset$ and refine it
iteratively. Refinements exist whenever there is a vertex that is not reached in $T$. Hence,
a pre-tree representation admits no further refinement exactly if it is a tree representation.

\vspace*{2em}%
\begin{figure}[ht]
  \begin{multicols*}{3}
    \begin{center}
      \scalebox{1.0}{
        \begin{tikzpicture}
          \tikzset{>=latex}
          \node[] (1) at (0,0) {$\natom{C}{1}$};
          \node[] (2) at (-1,-1) {$\natom{C}{2}$};
          \node[] (3) at (0,-2) {$\natom{C}{3}$};
          \node[] (4) at (1,-1) {$\natom{C}{4}$};
          \node[] (5) at (2,-2) {$\natom{N}{5}$};
          \node[] (6) at (0,-3) {$\natom{O}{6}$};

          \draw[thick, dotted, strongorange,double distance=0pt] (1)--(2);
          \draw[thick, dotted, strongorange,double distance=0pt] (2)--(3);
          \draw[thick, dotted, strongorange,double distance=0pt] (3)--(4);
          \draw[thick, dotted, strongorange,double distance=0pt] (3)--(6);
          \draw[thick, dotted, strongorange,double distance=0pt] (1)--(4);
          \draw[thick, dotted, strongorange,double distance=0pt] (4)--(5);
        \end{tikzpicture}
      }
    \end{center}
    \columnbreak
    \begin{center}
      \scalebox{1.0}{
        \begin{tikzpicture}
          \tikzset{>=latex}
          \node[] (1) at (0,0) {$\natom{C}{1}$};
          \node[] (2) at (-1,-1) {$\natom{C}{2}$};
          \node[] (3) at (0,-2) {$\natom{C}{3}$};
          \node[] (4) at (1,-1) {$\natom{C}{4}$};
          \node[] (5) at (2,-2) {$\natom{N}{5}$};
          \node[] (6) at (0,-3) {$\natom{O}{6}$};

          \draw[strongteal] (1)--(2);
          \draw[strongteal] (2)--(3);
          \draw[strongteal] (3)--(4);
          \draw[thick, dotted, strongorange,double distance=0pt] (3)--(6);
          \draw[thick, dotted, strongorange,double distance=0pt] (1)--(4);
          \draw[strongteal] (4)--(5);
        \end{tikzpicture}
      }
    \end{center}
    \columnbreak
    \begin{center}
      \scalebox{1.0}{
        \begin{tikzpicture}
          \tikzset{>=latex}
          \node[] (1) at (0,0) {$\natom{C}{1}$};
          \node[] (2) at (-1,-1) {$\natom{C}{2}$};
          \node[] (3) at (0,-2) {$\natom{C}{3}$};
          \node[] (4) at (1,-1) {$\natom{C}{4}$};
          \node[] (5) at (2,-2) {$\natom{N}{5}$};
          \node[] (6) at (0,-3) {$\natom{O}{6}$};

          \draw[strongteal] (1)--(2);
          \draw[strongteal] (2)--(3);
          \draw[strongteal] (3)--(4);
          \draw[strongteal] (3)--(6);
          \draw[thick, dotted, strongorange,double distance=0pt] (1)--(4);
          \draw[strongteal] (4)--(5);
        \end{tikzpicture}
      }
    \end{center}
  \end{multicols*}
  \caption{Refinement steps on a molecule with dotted edges in $C$ and solid edges in $T$}\label{fig:refinement}
\end{figure}
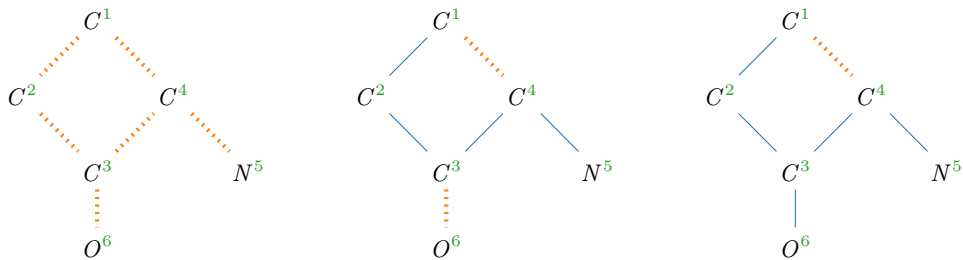
\begin{example}\label{example:refinement}
  Figure~\ref{fig:refinement} visualizes how a pre-tree representation
  is refined along the lines of Definition~\ref{def_refinement},
  starting with $T = \emptyset$ depicted on the left and
  arriving at a spanning tree in the rightmost depiction.
\end{example}

\begin{definition}\label{def_canon_graph}
A \emph{maximal refinement} of a molecular graph $G$ is a tree representation that is obtained from $G$
by a finite sequence of refinements. A \emph{centralized maximal refinement} is a tree representation
obtained from a maximal refinement by a permutation of vertices such that the root is central in its spanning tree (analogous to Definition~\ref{def_canon_tree}).
The canonical tree representation of a molecular graph $G$ is its $\prec$-largest centralized maximal refinement.
\end{definition}

\noindent In particular, the canonical tree representation coincides with the canonical molecular tree if $G$ is 
free of cycles.

\begin{figure}[ht]
  \begin{multicols*}{2}
    \begin{center}
      \scalebox{1.0}{
        \begin{tikzpicture}
          \tikzset{>=latex}
          \node[] (1) at (0,0) {$\natom{C}{1}$};
          \node[] (2) at (-1,-1) {$\natom{C}{2}$};
          \node[] (3) at (0,-2) {$\natom{C}{3}$};
          \node[circle,draw=red,inner sep=0,minimum size=.6cm] (_3) at (0,-2) {};
          \node[] (4) at (1,-1) {$\natom{C}{4}$};
          \node[] (5) at (2,-2) {$\natom{N}{5}$};
          \node[] (6) at (0,-3) {$\natom{O}{6}$};

          \draw[strongteal] (1)--(2);
          \draw[strongteal] (2)--(3);
          \draw[strongteal] (3)--(4);
          \draw[strongteal] (3)--(6);
          \draw[thick, dotted, strongorange,double distance=0pt] (1)--(4);
          \draw[strongteal] (4)--(5);
        \end{tikzpicture}
      }
    \end{center}
    \columnbreak
    \begin{center}
      \begin{tikzpicture}
          \tikzset{>=latex}
          \graph[grow down=0.65cm, branch right=0.45cm, nodes={circle,inner sep=0,minimum size=0cm}]{
            c1/$\natom{C}{1}$ [xshift=3.0cm, draw=red] --[strongteal] {
              c2/$\natom{C}{2}$ [yshift=-.4cm,xshift=2cm] --[strongteal] c3/$\natom{C}{3}$ [yshift=-.8cm,xshift=1cm],
              c4/$\natom{C}{4}$ [yshift=-.4cm,xshift=2.5cm] --[strongteal] n5/$\natom{N}{5}$ [yshift=-.8cm,xshift=2.45cm],
              o6/$\natom{O}{6}$ [yshift=-.4cm,xshift=3.2cm],
            };
            c3/$\natom{C}{3}$ --[thick,bend right, dotted, strongorange,double distance=0pt] c4/$\natom{C}{4}$;
          };
        \end{tikzpicture}
    \end{center}
  \end{multicols*}
  \vspace*{-1em}%
  \caption{Centralized maximal refinement (right) of the maximal refinement (left)}\label{fig:centr-max-refinement}
\end{figure}
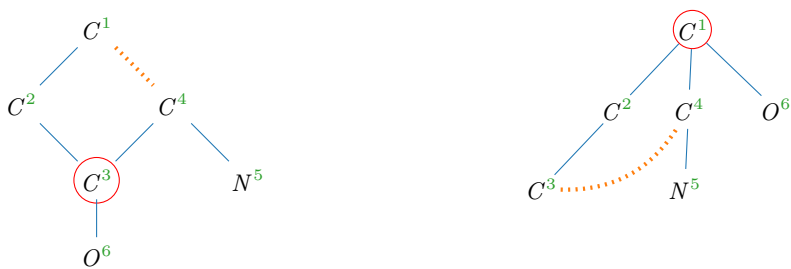
\vspace*{-1em}%
\begin{example}
  Figure~\ref{fig:centr-max-refinement} shows how the maximal refinement from Example~\ref{example:refinement}
  can be centralized on the circled vertex, permutating labels $\natom{C}{3}$ and $\natom{C}{1}$.
\end{example}

 \newcounter{linecounter}\newcounter{lastlinecounter}\setcounter{linecounter}{1}\newcommand\ASPlineInc[1]{\setcounter{lastlinecounter}{\value{linecounter}}\setcounter{linecounter}{\value{linecounter} + #1}}\newcounter{tmplinecounter}\newcommand\ASPline[1][1]{\setcounter{tmplinecounter}{\value{linecounter} + #1 - 1}\thetmplinecounter }

\newcommand\ASPLastline[1][1]{\setcounter{tmplinecounter}{\value{lastlinecounter} + #1 - 1}\thetmplinecounter }

\section{ASP Implementation}\label{sec-impl}

We compute the $\prec$-largest molecular graphs 
using Answer Set Programming (ASP),
a declarative logic programming language 
that is well-suited for combinatorial search problems.
Compared to a direct implementation,
ASP offers two main advantages: 
the solver efficiently handles traversal of the search space,
and ASP can serve as a domain-specific language to encode
additional constraints on the molecules.

\subsection{Answer Set Programming}

We give a brief introduction to Answer Set Programming in this section,
and direct interested readers to \cite[]{ASPIntro,ASPPrimer} for a comprehensive treatment.
ASP programs consist of facts, rules, and constraints that encode a problem,
with stable models corresponding to solutions.
To illustrate, consider the following simple ASP program:
\begin{ASPminted}[\thelinecounter]
ion(sodium). ion(chloride). ion(potassium).
1 { selected(I) : ion(I) } 2.
forms_salt :- selected(sodium), selected(chloride).  
:- selected(potassium), selected(chloride).
\end{ASPminted}
Line 2 of the above program declares three available ions as facts.
Line 4 contains a choice rule, 
selecting between one and two \ASPinline{ion} facts.
The rule on line 6 is read from right to left,
specifying that \ASPinline{from_salt} is derived 
if sodium and chloride were selected.
The constraint on line 8 filters out 
solutions that contain both potassium and chloride.

Evaluation of ASP programs follows the guess-and-check paradigm:
choice rules or negations are used to ``guess'' candidate solutions,
which are then tested by rules and constraints
to determine whether they constitude valid solutions.
For our example program,
we obtain the following solutions: 
\begin{align*}
&\{ \text{\ASPinline{ion(sodium)}} \},   
\{ \text{\ASPinline{ion(chloride)}}\},   
\{ \text{\ASPinline{ion(potassium)}} \} \\
&\{ \text{\ASPinline{ion(sodium)}},\, \text{\ASPinline{ion(potassium)}} \} \\
&\{ \text{\ASPinline{ion(sodium)}},\, \text{\ASPinline{ion(chloride)}},\, \text{\ASPinline{forms_salt}} \}
\end{align*}

Although the order of rule appearance in ASP programs is of no consequence semantically,
their logical interdependence prescribes an intuitive sequence of consequences.
For ease of presentation, we therefore may use vocabulary such as ``first'' and ``then''
in description of such programs where suitable.

Modern solvers like clingo \cite[]{clingo} in addition allow for more advanced features,
of which we make occasional use in our encoding.
This includes aggregates like
\ASPinline{M = }\begin{tikzpicture}
    \clip (-63pt,-4pt) rectangle (70pt,4pt);
    \node [align=center] (text) {\ASPinline{#min { X : condition(X) }}};
\end{tikzpicture},
which compute e.g. the minimum in a set of values,
syntactic sugar like range notation \ASPinline{fact(1..N)}, 
which expands to \ASPinline{fact(1)}, \dots, \ASPinline{fact(N)},
and Python script blocks that allow us to define custom functions.

\subsection{Application to Molecular Graph Generation}

Our implementation incorporates many of the conditions on canonical 
tree representations in the rules that infer these structures, rather than
relying on constraints to filter redundant representations later.
This is in contrast to other known approaches \cite[]{GJR20} for tree generation,
which first guess an edge relation and then prune inadequate graphs,
thereby producing many isomorphic results.
For molecular trees, our implementation achieves full symmetry-breaking along
the lines of Definition~\ref{def_canon_tree}. For graphs with cycles,
we merely approximate the conditions from Definition~\ref{def_canon_graph},
since the required $\prec$-maximality in this case seems to require a computationally
prohibitive search in ASP.\footnote{Section~\ref{sec_experiments} shows that the reduction in symmetry is still significant.}
It proceeds in the following steps:
\begin{enumerate}
    \item Choosing main-chain length, number of multi-bonds and cycle edges
    \item Distributing the element symbols, multi-bonds and cycle-edges on the vertices
    \item Building the spanning tree
    \item Comparing sibling subtrees and ensure they are $\prec$-decreasing from left to right
    \item Pruning representations that cannot be $\prec$-maximal due to bad choice of cycle edges
\end{enumerate}

As in Definition~\ref{def:tree-representation}, our implementation identifies vertices with
integers, with \ASPinline{1} being the root. Input molecular formulas are encoded in
facts \ASPinline{molecular_formula(|$e$|,|$f(e)$|)} and \ASPinline{element(|$e$|,|$n_e$|,|$\mathbb{V}(e)$|)},
for elements $e\in\mathbb{E}$ with atomic number $n_e$ (for usual ordering) and valence $\mathbb{V}(e)$.
We encode tree representations $G = \tuple{V, T \cup C, \ell, b}$
by facts \ASPinline{atom(|$v$|)} and \ASPinline{symbol(|$v$|, |$\ell(v)$|)} for $v \in V$;
\ASPinline{parent(|$v$|, |$i$|, |$v'$|)} for $\{v, v'\}\in T$ where $v'$ is the $(i+1)$th child of $v$;
and \ASPinline{cycle_start(|$v$|, |$c$|)} and \ASPinline{cycle_end(|$v'$|, |$c$|)} for $\{ v, v' \} \in C$ with $v < v'$,
where $c$ is a unique integer id for this edge.
Multiplicities of bonds $b(\{ v, v' \})$ are encoded only if $b(\{ v, v' \})>1$:
for $\{ v, v' \} \in T$, we associate bonds with the child $w=\max \{ v, v' \}$ 
using \ASPinline{multi_bond(|$w$|, |$b(\{ v, v' \})$|)}, 
whereas for $\{ v, v' \} \in C$, 
we encode $b(\{ v, v' \})$ 
single-bond cycles for the same $\{ v, v' \}$, which showed better performance in this case.
See \ref{app:pred-list} for a tabular overview of all predicates used in our encoding.

\begin{example}[Continuation of Example~\ref{example:problem}]
    The encoding for the sum formula $\textit{C}_5\textit{H}_5\textit{N}_5$ would look like: \\
    \vphantom{|}\indent{\footnotesize
    \ASPinline{molecular_formula("C",5).}
    \ASPinline{molecular_formula("H",5).}
    \ASPinline{molecular_formula("N",5).}} \\
    It necessitates the following valence information:\\
    \vphantom{|}\indent{\footnotesize
    \ASPinline{element("C", 6, 4).}
    \ASPinline{element("H", 1, 1).}
    \ASPinline{element("N", 7, 3).}} \\
    Matching this sum formula, Adenine (see also Figure~\ref{fig-example-adenine}) would be encoded like so:\\
\indent\begin{minipage}{.7\textwidth}
    \vphantom{|}\footnotesize
    \setlength{\tabcolsep}{1.5pt}
    \begin{tabular}{@{\indent}llll@{}}
        \ASPinline{symbol( 1,"C").} & & &
        \ASPinline{cycle_start(1,1).} \\
        \ASPinline{symbol( 2,"N").} &
        \ASPinline{parent(1,0, 2).} & & \\
        \ASPinline{symbol( 3,"C").} &
        \ASPinline{parent(2,0, 3).} &
        \ASPinline{multi_bond(3,2).} & \\
        \ASPinline{symbol( 4,"N").} &
        \ASPinline{parent(3,0, 4).} & & \\
        \ASPinline{symbol( 5,"C").} &
        \ASPinline{parent(4,0, 5).} &
        \ASPinline{multi_bond(5,2).} &
        \ASPinline{cycle_start(5,2).} \\
        \ASPinline{symbol( 6,"N").} &
        \ASPinline{parent(5,0, 6).} & & \\
        \ASPinline{symbol( 7,"C").} &
        \ASPinline{parent(1,1, 7).} &
        \ASPinline{multi_bond(7,2).} &
        \ASPinline{cycle_end(7,2).} \\
        \ASPinline{symbol( 8,"N").} &
        \ASPinline{parent(7,0, 8).} & & \\
        \ASPinline{symbol( 9,"C").} &
        \ASPinline{parent(8,0, 9).} &
        \ASPinline{multi_bond(9,2).} & \\
        \ASPinline{symbol(10,"N").} &
        \ASPinline{parent(9,0,10).} & &
        \ASPinline{cycle_end(10,1).} \\
    \end{tabular}
    \end{minipage}
    \begin{minipage}{.26\textwidth}
    \begin{center}
      \scalebox{0.7}{
        \begin{tikzpicture}
          \tikzset{>=latex}
          \graph[grow down=0.65cm, branch right=0.45cm, nodes={circle,inner sep=0,minimum size=0cm}]{
            c1/$\natom{C}{1}$ [xshift=2.0cm] --[strongteal] {
              n2/$\natom{N}{2}$ [xshift=1.6cm] --[double, strongteal] c3/$\natom{C}{3}$ [xshift=1.2cm] --[strongteal] n4/$\natom{N}{4}$ [xshift=0.8cm] --[double, strongteal] c5/$\natom{C}{5}$ [xshift=0.4cm] --[strongteal] n6/$\natom{N}{6}$,
              c7/$\natom{C}{7}$ [xshift=2.0cm, > double] --[strongteal] n8/$\natom{N}{8}$ [xshift=2.4cm] --[double, strongteal] c9/$\natom{C}{9}$ [xshift=2.8cm] --[strongteal] n10/$\natom{N}{10}$ [xshift=3.2cm]
            };
c5/$\natom{C}{5}$ --[thick,bend right, dotted, strongorange,double distance=0pt] c7/$\natom{C}{7}$;
            c1/$\natom{C}{1}$ --[thick,bend left, dotted, strongorange,double distance=0pt] n10/$\natom{N}{10}$;
          };
        \end{tikzpicture}
      }
   \end{center}
    \end{minipage}
\end{example}

Given the input, we guess facts for \ASPinline{symbol},
\ASPinline{multi_bond}, \ASPinline{cycle_start}, 
as well as \ASPinline{cycle_end}.
For efficiency, we avoid aggregates and instead proceed iteratively, updating counters as we make
guesses. We constrain possible guesses based on Definition~\ref{def:validgraph}.
For example, given a molecular formula $f$, the number of ``additional'' bonds used in cycles and multi-bonds
\begin{equation*}
    |C| + \sum_{\{v, v'\} \in T \cup C} \big( b(\{v, v'\})-1 \big)
\end{equation*}
known as the \emph{degree of unsaturation} in chemistry,
can be computed as
\begin{equation*}
    1 + \frac{1}{2} \sum_{e \in \mathbb{E}} f(e) \cdot (\mathbb{V}(e) - 2). 
\end{equation*}

When guessing multi-bonds and cycles, we therefore ensure that this number is met.
Multi-bonds \ASPinline{multi_bond(|$v$|,2)} or \ASPinline{multi_bond(|$v$|,3)} 
are guessed for non-root vertices $v$.\footnote{Higher bond multiplicities are not implemented in our prototype.}
For cycle markers, first we guess the number of \ASPinline{cycle_start}s at each vertex
and thereafter generate these facts with their unique cycle ids. Second, we guess the number of \ASPinline{cycle_end}s
at each vertex, making sure to not exceed the total count of \ASPinline{cycle_start}s at smaller vertices.
Using additional constraints, we ensure that each cycle has a single end,
start vertices are always smaller than end vertices, cycles never span a single edge (which should be represented as a multi-bond instead),
and two cycle edges with the same start are indexed according to their end vertex.

This completes the initial guessing phase for $\ell$, $b$, and $C$.
Facts that have the form \ASPinline{preset_bonds(|$v$|,|$\opfont{pre}(v)$|)} store the number of
bonding places $\opfont{pre}(v)$ that have been used up in the process for vertex $v$.
In the next phase, the program specifies possible spanning trees to establish Definition~\ref{def:validgraph}~\ref{it_validgraph_connect}.
Choices are limited since we aim at $\prec$-maximal tree representations, e.g., 
the subtree depth cannot increase from left to right.

We first guess the length of the longest path in the tree representation, whose central
elements are the only possible roots by Definition~\ref{def_canon_graph}.
This length is encoded as \ASPinline{main_chain_len(|$\mathit{length}$|)}.
It ranges from $1$ to $|V|$, but performance is gained by a better lower bound estimate:

\begin{ASPminted}[\thelinecounter]
1{ main_chain_len(@min_main_chain_len(N)..N) }1 :- non_hydrogen_atom_count(N).
\end{ASPminted}
\ASPlineInc{1}

\begin{proposition}
    Let $f$ be a molecular formula with 
    $N \coloneqq \sum_{e \in \mathbb{E} \setminus \{H\}} f(e)$ 
    non-hydrogen atoms,
    and let $G$ be a valid molecular graph for $f$.
    Given that the maximal valence 
    is $X=\max_{e\in \mathbb{E}}{\mathbb{V}(e)}$, $X > 2$, 
    a longest simple path in $G$
is at least of length
    \begin{tequation}\label{eq:min-main-chain-len}
        \min\left\{2 \cdot \left\lceil \log_{X-1}\left((X-2)\cdot \frac{N - 1}{X}+1\right)\right\rceil + 1,
                   2 \cdot \left\lceil \log_{X-1}\left((X-2)\cdot \frac{N}{2} + 1\right)\right\rceil \right\}
    \end{tequation}
\end{proposition}
\begin{proof}
    Consider a tree $G'$ with $N$ vertices where all nodes except the leaves have $X$ neighbors and let $l$ be the length of the longest simple paths in $G'$.
    A longest simple path in $G$ is at least of length $l$, as it is only increased by cycles, multi-bonds and lower-valence atoms.
    The two expressions in the $\min\{\cdot,\cdot\}$ estimate the minimal lengths
    in the case that $l$ is odd or even, respectively.

    In case $l$ is odd, $G'$ is a complete tree where the root has $X$ children and inner non-root vertices have $X-1$ children.
    Let $s_n^{odd}$ be the number of vertices in $G'$ if it has depth $n$:
    \begin{talign}\label{eq:sn_odd}
        s_n^{odd} = 1 + X \cdot \frac{(X-1)^{n-1}-1}{(X-2)}
    \end{talign}
    There is one root vertex which has $X$ subtrees of depth $n-1$.
    The longest path in this tree traverses first $n-1$ vertices in one of these $X$ subtrees,
    then the root vertex, and finally another $n-1$ vertices in another child tree,
    totalling $2 \cdot ( n-1 ) + 1$ nodes.\\
    Rearranging Equation~\ref{eq:sn_odd} to $2 \cdot ( n-1 )  +1$,
    we get $2 \cdot \left( \log_{X-1}\left((X-2)\cdot \frac{s_n^{odd} - 1}{X}+1\right)\right) + 1$.

    In case $l$ is even, $G'$ is combined of two equal trees where inner vertices have $X-1$ children.
    Let $s_n^{even}$ be the number of vertices in $G'$ if both trees have depth $n$:
    \begin{talign}\label{eq:sn_even}
        s_n^{even} = 2 \cdot \frac{(X-1)^{n}-1}{(X-2)}
    \end{talign}
    The longest path in this tree traverses first $n$ vertices in one and then further $n$ vertices in the other tree.
    Rearranging Equation~\ref{eq:sn_even} to $2 \cdot n $ yields $2 \cdot \left( \log_{X-1}\left((X-2)\cdot \frac{s_n^{even}}{2} + 1\right)\right)$.

    As the logarithm may give fractional values,
    both terms have to be rounded up: the minimal length of a longest path in $G$ is the next integer.
\end{proof}

We use a Python \texttt{\#script}-block during grounding to compute this lower bound procedurally (as this is not natively feasible in ASP).

Next, we iteratively guess \ASPinline{depth(|$v$|,|$d$|)}, \ASPinline{size(|$v$|,|$s$|)}, and \ASPinline{branching(|$v$|,|$b$|)} for $v\in V$,
where $b$ is the number of children of $v$, and $d$ and $s$ are the depth and size of the subtree with root $v$.
We require $d\leq s$ and $1\leq b\leq(\mathbb{V}(\ell(v))-\opfont{pre}(v))$.
Moreover, if $d>1$ then $b\geq 1$, and $b\geq 2$ for the root \ASPinline{1} unless $\left|V\right|\leq 2$.
The rules for \ASPinline{branching} are:

\vspace*{1em}%
\begin{ASPminted}[\thelinecounter]
branching(1, 1) :- non_hydrogen_atom_count(2).
1{ branching(1, 2..MAX) }1
    :- not branching(1, 1), symbol(1, E), element(E, _, VALENCE),
       MAX = #min{ N-1 : non_hydrogen_atom_count(N);
                   V-B : V=VALENCE, preset_bonds(1, B) }.
\end{ASPminted}
\ASPlineInc{5}\begin{ASPminted}[\thelinecounter]
1{ branching(I, 1..MAX) }1
    :- symbol(I, E), I>1, element(E, _, VALENCE),
       MAX = #min{ S-D+1 : S=SIZE, D=DEPTH;
                     V-B : V=VALENCE, preset_bonds(I, B) },
       size(I, SIZE), SIZE >= DEPTH, depth(I, DEPTH), DEPTH > 1.
\end{ASPminted}
\ASPlineInc{5}

At this point, the used-up binding places due to \ASPinline{multi_bond}s at child vertices
are captured in a fact \ASPinline{postset_bonds(|$v$|,|$\opfont{post}(v)$|)}.
Definition~\ref{def:validgraph}~\ref{it_validgraph_valence} is equivalent to a check of $\opfont{pre}(v) + \opfont{post}(v) \leq \mathbb{V}(\ell(v))$ for each $v \in V$.

Next, we split the main chain evenly between the first two children of root \ASPinline{1},
which have indices $2$ and $2+\opfont{size}(2)$. If the length is odd, the first child's depth is greater by $1$:

\begin{ASPminted}[\thelinecounter]
depth(2, ((MAIN_CHAIN_LEN-1)+(MAIN_CHAIN_LEN-1)\2)/2)
    :- main_chain_len(MAIN_CHAIN_LEN), MAIN_CHAIN_LEN > 1.
depth(2+LEFT_SIZE, ((MAIN_CHAIN_LEN-1)-(MAIN_CHAIN_LEN-1)\2)/2)
    :- main_chain_len(MAIN_CHAIN_LEN), MAIN_CHAIN_LEN > 2,
       size(2, LEFT_SIZE).
\end{ASPminted}
\ASPlineInc{5}

In general, the depth of a first child is always set to its parent's depth minus $1$. Depths for further children are chosen iteratively to be non-increasing.

\begin{ASPminted}[\thelinecounter]
depth(I+1, DEPTH-1)
    :- depth(I, DEPTH), atom(I+1), branching(I, _), DEPTH > 1.
1{ depth(POS_2, 1..PREV_DEPTH) }1
    :- branching(I, BRANCHING), BRANCHING > CHAIN,
       depth(POS_1, PREV_DEPTH),
       parent(I, CHILD_NR, POS_1), parent(I, CHILD_NR+1, POS_2).
\end{ASPminted}
\ASPlineInc{6}

The first child of a non-final vertex $v$ is always $v+1$ (line \ASPline[1] below).
Vertex ids for further children are chosen iteratively such that their left neighbor can reach its
depth and the parent's size is not exceeded (lines \ASPline[2]--\ASPline[7]).
These choices also determine the \ASPinline{size} of each child (not shown).

\begin{ASPminted}[\thelinecounter]
parent(I, 0, I+1) :- branching(I, _), non_hydrogen_atom_count(N), I < N.
1{ parent(I, CHILD_NR, SUM+DEPTH..MAX_CHILD) }1
    :- parent(I, CHILD_NR-1, SUM), depth(SUM, DEPTH), size(I, PARENT_S),
       branching(I, BRANCHING), BRANCHING > CHILD_NR,
       MAX_CHILD = #min{ N : non_hydrogen_atom_count(N);
                         T : T=I+PARENT_S-BRANCHING+CHILD_NR },
       MAX_CHILD >= SUM+DEPTH.
\end{ASPminted}
\ASPlineInc{7}

Next, we materialize the total order $\prec$ from Section~\ref{sec_canon_tree}
in a predicate \ASPinline{lt}. For graphs with cycles, we use the number of cycle markers
per vertex as an additional ordering criterion instead of the (more costly)
tree transformation of Section~\ref{sec_canon_graph}.
The following constraints exclude cases that cannot be $\prec$-maximal,
due to children traversed in $\prec$-increasing order (line \ASPline[1]) or
choice of a non-optimal central vertex as root (line \ASPline[2]).

\begin{ASPminted}[\thelinecounter]
:- parent(I, CHILD_NR, I1), parent(I, CHILD_NR+1, I2), lt(I1, I2).
:- main_chain_len(MAIN_CHAIN_LEN), MAIN_CHAIN_LEN\2 = 0, lt(1, 2).
\end{ASPminted}
\ASPlineInc{2}

At this point, perfect symmetry-breaking for acyclic graphs has been achieved.
Cyclic graphs, however, can still have isomorphic representations,
since the implementation (a) does not compare all possible choices of main chain, and
(b) does not ensure that tree edges are obtained from longest extensions
as in Definition~\ref{def_refinement}.
For (b), repeated longest path computations are impractical, but we can
heuristically eliminate many non-optimal choices by excluding obvious violations.
\begin{definition}\label{def:shortening-cycle-edge}
    Let $G=\tuple{V,T\cup C,\ell,b}$ be a tree representation with cycle edge $e=\{v_1,v_2\}$.
    Let $P=v_1, \ldots, v_2$ be the unique path in $G$ that consists only of tree edges.
    We say that $e$ is \textit{shortening}, if an $e' \in P$ exists, s.t.
    $G'=\tuple{V,T'\cup C',\ell,b}$ with $T'=\left(T \setminus \{e'\}\right) \cup \{e\}$ and $C'=\left(C \setminus \{e\}\right) \cup \{e'\}$ is deeper.
\end{definition}
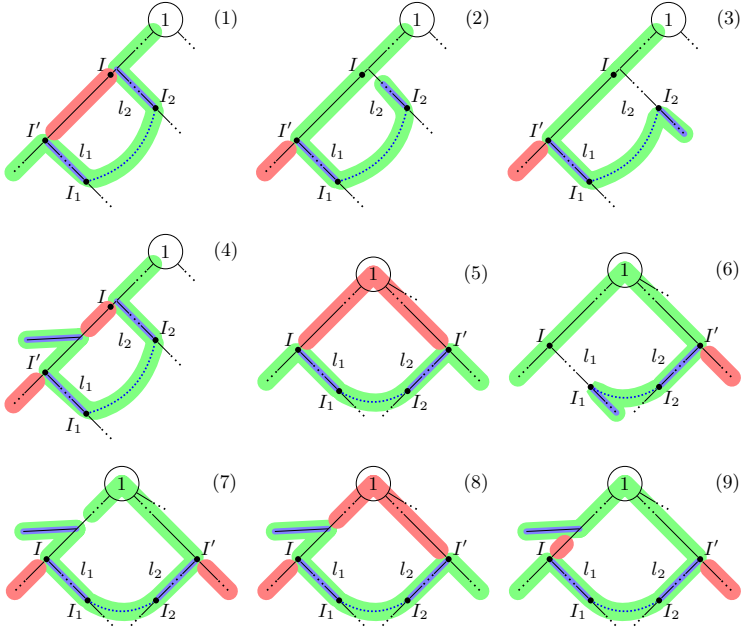
\begin{figure}
    \centering
    \scalebox{.8}{\begin{minipage}{\textwidth}\centering\tikzset{placeholder/.style={minimum size=0, inner sep=0},
         place/.style args={#1 of #2}{below #1=0.17cm and 0.17cm of #2, placeholder},
         dot/.style={minimum size=0.1cm, inner sep=0, fill, circle},
         etc/.style={dotted, thick},
         highlight/.style={rounded corners, line cap=round, line width=0.3cm},
         innerhighlight/.style={rounded corners, line cap=round, line width=0.1cm}}

\pgfdeclarelayer{bg}    \pgfsetlayers{bg,main}  
\newenvironment*{CycleImg}[1][]{
\begin{tikzpicture}
    \node[] at (1cm,0) {(#1)};

    \node[draw, circle] (q1) {$1$};
    \node[below left=0.10cm and 0.10cm of q1, placeholder] (q2) {};
    \node[place=left of q2] (q3) {};
    \node[below left=0.30cm and 0.30cm of q3, placeholder] (q4) {};
    \node[below left=0.05cm and 0.05cm of q4, dot] (q5) {};
    \node[above left=0.02cm of q5, inner sep=0] (q5_label) {$I$};
    \node[below left=1.00cm and 1.00cm of q5, dot] (q6) {};
    \node[above left=0.02cm of q6, inner sep=0] (q6_label) {$I'$};
    \node[below left=0.30cm and 0.30cm of q6, placeholder] (q7) {};
    \node[place=left of q7] (q8) {};

    \node[below right=0.20cm and 0.20cm of q6, placeholder] (q9) {};
    \node[place=right of q9] (q10) {};
    \node[above right=0.01cm of q10] (q10_label) {$l_1$};
    \node[below right=0.20cm and 0.20cm of q10, dot] (q11) {};
    \node[below left=0.02cm of q11, inner sep=0] (q11_label) {$I_1$};
    \node[below right=0.20cm and 0.20cm of q11, placeholder] (q12) {};
    \node[place=right of q12] (q13) {};

    \node[below right=0.20cm and 0.20cm of q4, placeholder] (q14) {};
    \node[place=right of q14] (q15) {};
    \node[below left=0.01cm of q15] (q15_label) {$l_2$};
    \node[below right=0.20cm and 0.20cm of q15, dot] (q16) {};
    \node[above right=0.02cm of q16, inner sep=0] (q16_label) {$I_2$};
    \node[below right=0.20cm and 0.20cm of q16, placeholder] (q17) {};
    \node[place=right of q17] (q18) {};

    \node[below right=0.10cm and 0.10cm of q1, placeholder] (q19) {};
    \node[place=right of q19] (q20) {};

    \draw[] (q1) -- (q2);
    \draw[etc] (q2) -- (q3);
    \draw[] (q3) -- (q4);
    \draw[] (q4) -- (q5);
    \draw[] (q5) -- (q6);
    \draw[] (q6) -- (q7);
    \draw[etc] (q7) -- (q8);

    \draw[] (q6) -- (q9);
    \draw[etc] (q9) -- (q10);
    \draw[] (q10) -- (q11);
    \draw[] (q11) -- (q12);
    \draw[etc] (q12) -- (q13);

    \draw[] (q4) -- (q14);
    \draw[etc] (q14) -- (q15);
    \draw[] (q15) -- (q16);
    \draw[] (q16) -- (q17);
    \draw[etc] (q17) -- (q18);

    \draw[] (q1) -- (q19);
    \draw[etc] (q19) -- (q20);

    \draw[densely dotted, blue, thick] (q11) to[bend right] (q16);
}{
\end{tikzpicture}
}

\newenvironment*{CycleImgBg}[1][]{
\begin{CycleImg}[#1]
    \begin{pgfonlayer}{bg}
}{
    \end{pgfonlayer}
\end{CycleImg}
}

\newenvironment*{CrossCycleImg}[1][]{
\begin{tikzpicture}
    \node[] at (1.7cm,0) {(#1)};

    \node[draw, circle] (q1) {$1$};
    \node[placeholder] (q1_help) at (q1) {};
    \node[below left=0.10cm and 0.10cm of q1, placeholder] (q2) {};
    \node[place=left of q2] (q3) {};
    \node[below left=0.70cm and 0.70cm of q3, dot] (q4) {};
    \node[above left=0.02cm of q4, inner sep=0] (q4_label) {$I$};
    \node[below left=0.30cm and 0.30cm of q4, placeholder] (q5) {};
    \node[place=left of q5] (q6) {};

    \node[below right=0.20cm and 0.20cm of q4, placeholder] (q7) {};
    \node[place=right of q7] (q8) {};
    \node[above right=0.01cm of q8] (q8_label) {$l_1$};
    \node[below right=0.20cm and 0.20cm of q8, dot] (q9) {};
    \node[below left=0.02cm of q9, inner sep=0] (q9_label) {$I_1$};
    \node[below right=0.20cm and 0.20cm of q9, placeholder] (q10) {};
    \node[place=right of q10] (q11) {};

    \node[below right=0.10cm and 0.10cm of q1, placeholder] (q12) {};
    \node[place=right of q12] (q13) {};
    \node[below right=0.70cm and 0.70cm of q13, dot] (q14) {};
    \node[above right=0.02cm of q14, inner sep=0] (q14_label) {$I'$};
    \node[below right=0.30cm and 0.30cm of q14, placeholder] (q15) {};
    \node[place=right of q15] (q16) {};

    \node[below left=0.20cm and 0.20cm of q14, placeholder] (q17) {};
    \node[place=left of q17] (q18) {};
    \node[above left=0.01cm of q18] (q18_label) {$l_2$};
    \node[below left=0.20cm and 0.20cm of q18, dot] (q19) {};
    \node[below right=0.02cm of q19, inner sep=0] (q19_label) {$I_2$};
    \node[below left=0.20cm and 0.20cm of q19, placeholder] (q20) {};
    \node[place=left of q20] (q21) {};

    \node[below right=0.10cm and 0.30cm of q1, placeholder] (q22) {};
    \node[below right=0.10cm and 0.20cm of q22, placeholder] (q23) {};

    \draw[] (q1) -- (q2);
    \draw[etc] (q2) -- (q3);
    \draw[] (q3) -- (q4);
    \draw[] (q4) -- (q5);
    \draw[etc] (q5) -- (q6);

    \draw[] (q4) -- (q7);
    \draw[etc] (q7) -- (q8);
    \draw[] (q8) -- (q9);
    \draw[] (q9) -- (q10);
    \draw[etc] (q10) -- (q11);

    \draw[] (q1) -- (q12);
    \draw[etc] (q12) -- (q13);
    \draw[] (q13) -- (q14);
    \draw[] (q14) -- (q15);
    \draw[etc] (q15) -- (q16);

    \draw[] (q14) -- (q17);
    \draw[etc] (q17) -- (q18);
    \draw[] (q18) -- (q19);
    \draw[] (q19) -- (q20);
    \draw[etc] (q20) -- (q21);

    \draw[] (q1) -- (q22);
    \draw[etc] (q22) -- (q23);

    \draw[densely dotted, blue, thick] (q9) to[bend right] (q19);
}{
\end{tikzpicture}
}

\newenvironment*{CrossCycleImgBg}[1][]{
\begin{CrossCycleImg}[#1]
    \begin{pgfonlayer}{bg}
}{
    \end{pgfonlayer}
\end{CrossCycleImg}
}

\newenvironment*{CrossCycleImgSideBg}[1][]{
\begin{CrossCycleImg}[#1]
    \node[placeholder] (q_base) at ($(q4)+(0.5cm,0.5cm)$) {};
    \node[] (q_side) at ($(q_base)-(1.0cm,0.05cm)$) {};

    \draw[] (q_base) -- (q_side);

    \begin{pgfonlayer}{bg}
}{
    \end{pgfonlayer}
\end{CrossCycleImg}
}

\begin{subfigure}[c]{0.3\textwidth}
    \centering
\begin{CycleImgBg}[1]
    \path[draw=green!50, highlight] (q1) -- (q4) -- (q16) to[bend left] (q11) -- (q6) -- (q8);
    \path[draw=blue!50, innerhighlight] (q4) -- (q16);
    \path[draw=blue!50, innerhighlight] (q6) -- (q11);
    \path[draw=red!50, highlight] (q5) -- ($(q6)+(0.15cm,0.15cm)$);
\end{CycleImgBg}
\end{subfigure}
\begin{subfigure}[c]{0.3\textwidth}
    \centering
\begin{CycleImgBg}[2]
    \path[draw=green!50, highlight] (q1) -- (q4) -- (q6) -- (q11) to[bend right] (q16) -- ($(q4)+(0.25cm,-0.25cm)$);
    \path[draw=blue!50, innerhighlight] ($(q4)+(0.25cm,-0.25cm)$) -- (q16);
    \path[draw=blue!50, innerhighlight] (q6) -- (q11);
    \path[draw=red!50, highlight] ($(q6)-(0.15cm,0.15cm)$) -- (q8);
\end{CycleImgBg}
\end{subfigure}
\begin{subfigure}[c]{0.3\textwidth}
    \centering
\begin{CycleImgBg}[3]
    \path[draw=green!50, highlight] (q1) -- (q4) -- (q6) -- (q11) to[bend right] (q16) -- (q18);
    \path[draw=blue!50, innerhighlight] (q16) -- (q18);
    \path[draw=blue!50, innerhighlight] (q6) -- (q11);
    \path[draw=red!50, highlight] ($(q6)-(0.15cm,0.15cm)$) -- (q8);
\end{CycleImgBg}
\end{subfigure}
\begin{subfigure}[c]{0.3\textwidth}
    \centering
\begin{CycleImg}[4]
    \node[placeholder] (q_base) at ($(q5)-(0.5cm,0.5cm)$) {};
    \node[] (q_side) at ($(q_base)-(1.0cm,0.05cm)$) {};

    \draw[] (q_base) -- (q_side);

    \begin{pgfonlayer}{bg}
        \path[draw=green!50, highlight] (q1) -- (q4) -- (q16) to[bend left] (q11) -- (q6) -- (q_base) -- (q_side);
        \path[draw=blue!50, innerhighlight] (q4) -- (q16);
        \path[draw=blue!50, innerhighlight] (q6) -- (q11);
        \path[draw=blue!50, innerhighlight] (q_base) -- (q_side);
        \path[draw=red!50, highlight] ($(q6)-(0.15cm,0.15cm)$) -- (q8);
        \path[draw=red!50, highlight] ($(q4)-(0.15cm,0.15cm)$) -- ($(q_base)+(0.15cm,0.15cm)$);
    \end{pgfonlayer}
\end{CycleImg}
\end{subfigure}
\begin{subfigure}[c]{0.3\textwidth}
    \centering
\begin{CrossCycleImgBg}[5]
    \path[draw=green!50, highlight] (q6) -- (q4) -- (q9) to[bend right] (q19) -- (q14) -- (q16);
    \path[draw=blue!50, innerhighlight] (q4) -- (q9);
    \path[draw=blue!50, innerhighlight] (q14) -- (q19);
    \path[draw=red!50, highlight] ($(q4)+(0.15cm,0.15cm)$) -- (q1_help) -- ($(q14)-(0.15cm,-0.15cm)$);
\end{CrossCycleImgBg}
\end{subfigure}
\begin{subfigure}[c]{0.3\textwidth}
    \centering
\begin{CrossCycleImgBg}[6]
    \path[draw=green!50, highlight] (q6) -- (q1_help) -- (q14) -- (q19) to[bend left] (q9) -- (q11);
    \path[draw=blue!50, innerhighlight] (q9) -- (q11);
    \path[draw=blue!50, innerhighlight] (q14) -- (q19);
    \path[draw=red!50, highlight] ($(q14)+(0.15cm,-0.15cm)$) -- (q16);
\end{CrossCycleImgBg}
\end{subfigure}
\begin{subfigure}[c]{0.3\textwidth}
    \centering
\begin{CrossCycleImgSideBg}[7]
    \path[draw=green!50, highlight] (q_side) -- (q_base) -- (q4) -- (q9) to[bend right] (q19) -- (q14) -- (q1_help) -- ($(q_base)+(0.25cm,0.25cm)$);
    \path[draw=blue!50, innerhighlight] (q4) -- (q9);
    \path[draw=blue!50, innerhighlight] (q14) -- (q19);
    \path[draw=blue!50, innerhighlight] (q_base) -- (q_side);
    \path[draw=red!50, highlight] ($(q4)-(0.15cm,0.15cm)$) -- (q6);
    \path[draw=red!50, highlight] ($(q14)+(0.15cm,-0.15cm)$) -- (q16);
\end{CrossCycleImgSideBg}
\end{subfigure}
\begin{subfigure}[c]{0.3\textwidth}
    \centering
\begin{CrossCycleImgSideBg}[8]
    \path[draw=green!50, highlight] (q_side) -- (q_base) -- (q4) -- (q9) to[bend right] (q19) -- (q14) -- (q16);
    \path[draw=blue!50, innerhighlight] (q4) -- (q9);
    \path[draw=blue!50, innerhighlight] (q14) -- (q19);
    \path[draw=blue!50, innerhighlight] (q_base) -- (q_side);
    \path[draw=red!50, highlight] ($(q_base)+(0.15cm,0.15cm)$) -- (q1_help) -- ($(q14)-(0.15cm,-0.15cm)$);
    \path[draw=red!50, highlight] ($(q4)-(0.15cm,0.15cm)$) -- (q6);
\end{CrossCycleImgSideBg}
\end{subfigure}
\begin{subfigure}[c]{0.3\textwidth}
    \centering
\begin{CrossCycleImgSideBg}[9]
    \path[draw=green!50, highlight] (q_side) -- (q_base) -- (q1_help) -- (q14) -- (q19) to[bend left] (q9) -- (q4) -- (q6);
    \path[draw=blue!50, innerhighlight] (q4) -- (q9);
    \path[draw=blue!50, innerhighlight] (q14) -- (q19);
    \path[draw=blue!50, innerhighlight] (q_base) -- (q_side);
    \path[draw=red!50, highlight] ($(q4)+(0.15cm,0.15cm)$) -- ($(q_base)-(0.25cm,0.25cm)$);
    \path[draw=red!50, highlight] ($(q14)+(0.15cm,-0.15cm)$) -- (q16);
\end{CrossCycleImgSideBg}
\end{subfigure}
 \end{minipage}}
    \caption{Patterns for detecting shortening cycles}
    \label{fig:cycle-heuristic}
\end{figure}
Note that for any graph~$G$, 
its canonical tree representation $\max\limits_\prec{[G]_{\cong}}$
cannot contain shortening cycles.
Otherwise, one could construct a representation with greater depth,
which would therefore be $\prec$-larger.

Our implementation detects shortening cycles in a tree representation
by matching it against the patterns depicted in Figure~\ref{fig:cycle-heuristic}.
Note that this listing is not necessarily exhaustive. 
In green, we mark paths that use a cycle edge connecting 
$I_1$ and $I_2$,
which would increase the depth of the node $I$ in the cases (1) -- (4),
or would increase the length of the main chain in the cases (5) -- (9).
For instance, in case (1) this is detected by computing the lengths $l_1$ (between node $I'$ and $I_1$) 
and $l_2$ (between $I$ and $I_2$), and comparing them with the length between $I$ and $I'$.
This computation is immediate
since the respective nodes lie on the same branch, 
whose nodes are numbered consecutively
according to the depth-first labeling required by Defition~\ref{def:tree-representation}.
The extended path created by the cycle edge is represented in blue, 
and the nodes that are no longer part of the path are colored red.

\subsection{Generalization to Graph Problems}

While our encoding is optimized for the generation of molecules,
many insights of our implementation also apply in the general case
for generating any undirected graph, optionally with degree-constraints.
Hence, we also provide ASP programs that
\begin{enumerate}
    \item enumerate all graphs that satisfy prescribed degree specifications for a given set of nodes.\footnote{https://github.com/knowsys/eval-2024-asp-molecules/blob/main/graph\_degree.lp}\label{lab:generalized_degree}
    \item enumerate all undirected graphs for a given number of nodes,\footnote{https://github.com/knowsys/eval-2024-asp-molecules/blob/main/graph.lp}\label{lab:generalized_graph}
\end{enumerate} 
The ASP program in item~\ref{lab:generalized_degree} 
takes predicates like \ASPinline{nodes_with_deg(DEGREE, COUNT)} as input
instead of \ASPinline{molecular_formula/2} and \ASPinline{element/3}.
We guess \ASPinline{intended_degree/2} facts for each node, 
rather than assigning elements through \ASPinline{symbol/2}, 
which previously determined node degrees.
Since multi-edges are no longer considered, 
there is no need to guess \ASPinline{multi_bond} facts.
The degree of unsaturation is still 
used to infer the number of cycle edges that must be distributed.
In addition, the definition of the \ASPinline{lt} 
predicate was adapted to use \ASPinline{intended_degree} 
instead of \ASPinline{symbol}, and to omit checks related to multi-bonds.

To generate undirected graph structures 
in item~\ref{lab:generalized_graph}, 
we first guess the degree of each node and then proceed analogously to the previous case. \section{Experimental Evaluation}\label{sec_experiments}

We evaluate our ASP implementation (``\textsc{Genmol}'') for correctness,
avoidance of redundant solutions, and runtime.
All of our experiments were conducted on a
mid-end server (2$\times$QuadCore Intel Xeon 3.5GHz, 768GiB RAM, Linux NixOS 23.11)
using clingo v5.7.1 for ASP reasoning.
Evaluation data, scripts, and results are available online at
\url{https://github.com/knowsys/eval-2024-asp-molecules}.

\paragraph*{Evaluated Systems.}
The ASP-based core of our system \textsc{Genmol} consists of 174 rules
(including 44 constraints).\footnote{\url{https://github.com/knowsys/eval-2024-asp-molecules/blob/main/smiles.lp}}
As a gold standard, we use the
existing commercial tool \textsc{Molgen} (\url{https://molgen.de}), which produces molecular graphs using a proprietary canonicalization approach.
Moreover, we compare our approach to four ASP-based solutions,
labeled \textsc{Naive}, \textsc{Graph}, \textsc{sbass}, and \textsc{BreakID}. 
\textsc{Naive} is a direct ASP encoding of Definition~\ref{def:validgraph}, which serves as a baseline:
\begin{ASPminted}[1]


{ edge(X, Y) : atom(X), atom(Y), X < Y }. edge(Y, X) :- edge(X, Y).

1{ edge(X, Y, 1..3) }1 :- edge(X, Y), X < Y. edge(Y, X, M) :- edge(X, Y, M).

degree(N, D) :- atom(N), D = #sum { C, X : edge(N, X, C) }.
reachable(1). reachable(Y) :- reachable(X), edge(X, Y).
:- not reachable(X), atom(X).

:- atom(N), symbol(N, E), degree(N, D), element(E, _, V), D > V.

:- DEGREE_SUM = #sum { D : degree(N, D) },
   VALENCE_SUM = #sum { E, N : symbol(N, E), element(E, _, V) },
   molecular_formula("H", H_COUNT),
   H_COUNT != VALENCE_SUM - DEGREE_SUM.
\end{ASPminted}

\noindent The above listing skips the construction of the 
\ASPinline{atom(|$v$|)} and \ASPinline{symbol(|$v$|, |$\ell(v)$|)},
which can easily be obtained from the \ASPinline{molecular_formula}.

\textsc{Graph} refines \textsc{Naive} with symmetry-breaking constraints
based on Definition~12 of \cite[]{graph}, 
which apply to partitioned simple graphs $G$
represented by their adjacency matrix $\mathcal{A}_G$:
\begin{talign}
  \text{sb}(G) = \bigwedge_{e \in \mathbb{E}}\; \bigwedge_{\substack{\ell(i) = \ell(j) = e, \\i < j,\, j - i \neq 2}} \mathcal{A}_G[i] \preceq_{\{i, j\}} \mathcal{A}_G[j].
  \label{eq_symmetry}
\end{talign}
Here, $\preceq_{\{i, j\}}$ denotes the lexicographic order
comparing the $i$th and $j$th row of the adjacency matrix $\mathcal{A}_G$ 
of a molecular graph $G$, ignoring columns $i$ and $j$.
Graph representations that do not satisfy (\ref{eq_symmetry}) are pruned.
These constraints can be succinctly represented in ASP,
and are appended to the naive implementation:
\footnote{\url{https://github.com/knowsys/eval-2024-asp-molecules/blob/main/lex.lp}}
A graph $G$ satisfies $\text{sb}(G)$
if the rows of its adjacency matrix are lexicographically ordered.
This means that every nonzero entry in the matrix must be followed
by rows that either:
\begin{enumerate}
  \item Contain a nonzero entry further to the left (see lines 1--3), or
  \item Have a larger entry in the same column (see lines 4--6).
\end{enumerate}  
This is captured by the auxiliary predicate \ASPinline{sat(I, K, J)},
which specifies whether rows \ASPinline{I} and \ASPinline{J} 
are in lexicographic order up to column \ASPinline{K},
as illustrated in Figure~\ref{fig:graph-sat}.
Note that we only compare rows, whose nodes are assigned to the same element.
\begin{ASPminted}[1]
sat(I, K, J) :-
  symbol(I, T), symbol(J, T), symbol(K, T), symbol(L, T), J > I, J - I != 2,
  edge(I, K), edge(J, L), L < K, L != I.
sat(I, K, J) :-
  symbol(I, T), symbol(J, T), symbol(K, T), J > I, J - I != 2,
  edge(I, K, N), edge(J, K, M), N <= M.

:- symbol(I, T), symbol(J, T), symbol(K, T),
   edge(I, K), not sat(I, K, J), J > I, K != J, J - I != 2.
\end{ASPminted}
\begin{figure}
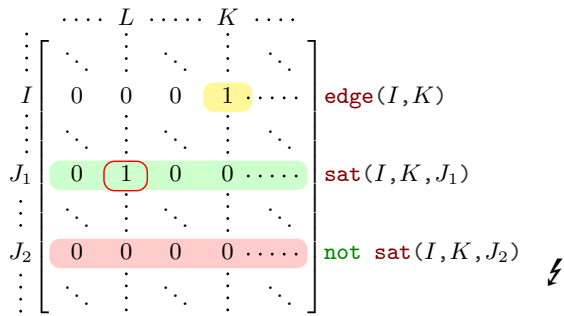

    \centering
    \newcommand\yf[2]{\Block[fill=#1,rounded-corners]{#2}{}}\newcommand\yd[2]{\Block[draw=#1,rounded-corners]{#2}{}}\newcommand\blap[1]{\vbox to 0pt{\hbox{#1}\vss}}\begin{equation*}
      \begin{bNiceMatrix}[first-row,first-col,last-col,margin,xdots/shorten=2mm]
               &                   \Cdots &                    L & \Cdots &                          K & \Cdots & \vphantom{.} \\
        \Vdots &                   \ddots &               \Vdots & \ddots &                     \Vdots & \ddots & \\
        I      &                        0 &                    0 &      0 & \yf{yellow!50}{1-1}      1 & \Cdots & \text{\ASPinline{edge(|$I$|,|$K$|)}} \\
        \Vdots &                   \ddots &               \Vdots & \ddots &                     \Vdots & \ddots & \\
        J_1    & \yf{green!20}{1-}      0 & \yd{red}{1-1}      1 &      0 &                          0 & \Cdots & \text{\ASPinline{sat(|$I$|,|$K$|,|$J_1$|)}} \\
        \Vdots &                   \ddots &               \Vdots & \ddots &                     \Vdots & \ddots & \\
        J_2    & \yf{red!20}{1-}        0 &                    0 &      0 &                          0 & \Cdots & \text{\ASPinline{not sat(|$I$|,|$K$|,|$J_2$|)}} $\quad \text{\blap{\Large\Lightning}}$ \\
        \Vdots &                   \ddots &               \Vdots & \ddots &                     \Vdots & \ddots & \\
      \end{bNiceMatrix}
    \end{equation*}
    \caption{Visualization of \textsc{Graph} symmetry-breaking encoding}
    \label{fig:graph-sat}
\end{figure}

\textsc{sbass} \cite[]{sbass} and \text{BreakID} \cite[]{breakid_system_description}
append instance-specific symmetry-breaking constraints to the ground program
by considering the automorphisms of its graph representation.
\text{BreakID} is run on the \textsc{Naive} implementation,
while we used \textsc{sbass} with an equivalent aggregate-free version.\footnote{\url{https://github.com/knowsys/eval-2024-asp-molecules/blob/main/naive-SBASS.lp}}

All evaluated ASP encodings were passed through 
the heuristic non-ground optimizer \textsc{ngo} (\url{https://potassco.org/ngo/}),
but no performance improvements were observed, 
suggesting that our encodings are reasonably efficient.

\paragraph*{Data set.}
For evaluation, we have extracted a dataset of molecules with
molecular formulas and graph structures using
the Wikidata SPARQL service \cite[]{Malyshev+2018:WikidataSPARQL}.
We selected 8,980 chemical compounds of up to 17 atoms (due to performance constraints)
with SMILES and an article on English Wikipedia
as a proxy for practical relevance.
Compounds with unconnected molecular graphs,
atoms of non-standard valence, and subgroup elements were excluded,
resulting in a dataset of 5,625 entries, of which we found 152 to have non-parsable SMILES.

\paragraph*{Evaluation of Correctness.}
Given the complexity of the implementation, we also assess its correctness empirically.
To this end, we augment our program with ASP rules that take an additional
direct encoding of a molecular graph as input, and that check if the
molecular graph found by \textsc{Genmol} is isomorphic to it.
This allows us to
determine if the given structures of molecules in our data set can be found
in our tool.
The validation graph structure is encoded in facts \ASPinline{required_bond(|$v_1$|,|$\ell(v_1)$|,|$v_2$|,|$\ell(v_2)$|,|$b(\{v_1,v_2\})$|)}
that were extracted from the SMILES representation in Wikidata.

Correctness experiments were measured with a timeout of 7 minutes.
Out of 5,473 compounds, a matching molecular structure was found for 5,338,
whereas 132 could not be processed within the timeout. For three compounds,
\href{https://www.wikidata.org/wiki/Q21099635}{\emph{Sandalore} (Wikidata ID Q21099635)}, and
\href{https://www.wikidata.org/wiki/Q113691142}{\emph{Eythrohydrobupropion} (Q113691142)} as well as
\href{https://www.wikidata.org/wiki/Q72518680}{\emph{Threodihydrobupropion} (Q72518680)},
the given structures could not be reproduced, which we traced back to errors in Wikidata that
we have subsequently corrected.

The evaluation therefore suggests that \textsc{Genmol} can find the correct molecular structures
across a wide range of actual compounds.
Timeouts occurred primarily for highly unsaturated, larger compounds (over 16 atoms),
where millions of solutions exist.

\paragraph*{Evaluation of symmetry-breaking.}
To assess to what extent our approximated implementation of canonical tree representations
succeeds in avoiding redundant isomorphic solutions, we consider the smallest 1,750
distinct molecular formulas from our data set.
We then computed molecular graph representations for all 1,750 cases for each of our evaluated systems,
using \textsc{Molgen} as a gold standard to determine the actual number of distinct molecular graphs.
The timeout for these experiments was 60 seconds.
\begin{figure}[h]
\includegraphics[width=\textwidth]{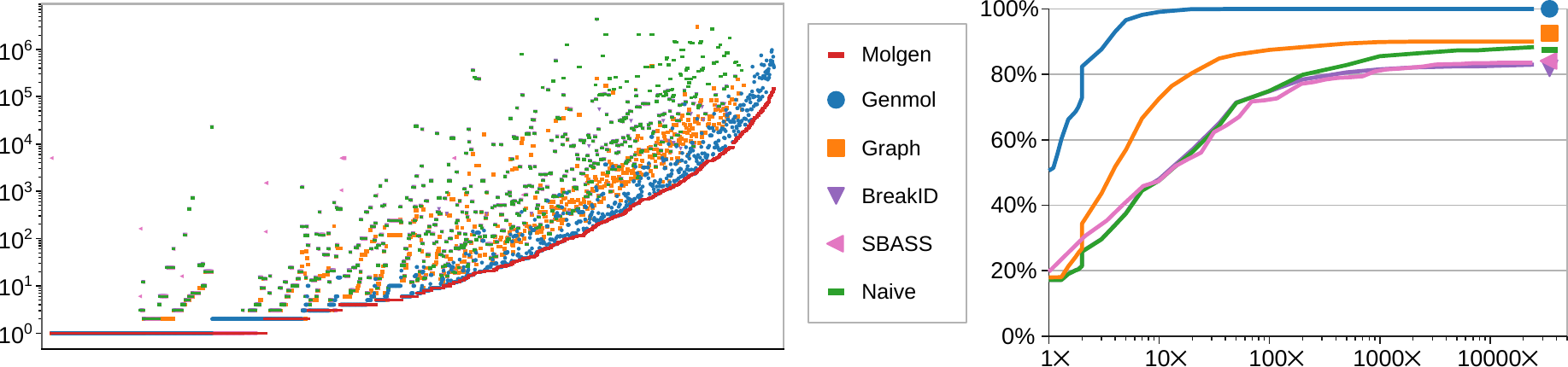}\caption{Number of models for each compound in the data set (left) and ratio of compounds with model counts within a factor of the gold standard (right)}\label{fig-comp-symmetry-breaking}
\end{figure}
The number of returned solutions are shown in Figure~\ref{fig-comp-symmetry-breaking} (left), with samples sorted
by their number of distinct graphs according to \textsc{Molgen}. As expected, \textsc{Molgen} is a lower
bound, and in particular no implementation finds fewer representations (which would be a concern for correctness),
while \textsc{Naive} is an upper bound. As expected, no ASP tool achieves perfect canonization of results, but the
difference between the number of solutions and the optimum vary significantly. In particular,
 \textsc{BreakID} rarely improves over \textsc{Naive} (just 24 such cases exist), though it does cause one third more timeouts.

For \textsc{Naive}, some samples led to over 20,000 times more models than \textsc{Molgen}, whereas the largest
such factor was just above $39$ for \textsc{Genmol} (for $\mathit{C}_8\mathit{H}_2$).
Figure~\ref{fig-comp-symmetry-breaking} (right) shows the ratio of samples with model counts
within a certain factor of the gold standard. For example, the values at $10$ show the ratios of samples for
which at most ten times as many models were computed than in \textsc{Molgen}: this is $99\%$ for \textsc{Genmol},
$72\%$ for \textsc{Graph}, and $48\%$ for \textsc{BreakID}, \textsc{sbass} and \textsc{Naive}.
All ratios refer to the same total, so the curves converge to the ratio of cases solved within 
the timeout.
Their starting point marks the ratio with exact model counts:
$51\%$ for \textsc{Genmol} and $17\%-19\%$ for the others.

We conclude that symmetry-breaking in \textsc{Genmol}, 
while not perfect, 
performs very well in comparison with generic approaches. 
In absolute terms, the results are often close enough to the optimum 
that any remaining redundancies could be removed in a post-processing 
step using conventional graph isomorphism tools. 
Although applying such checks to the entire solution space 
would be computationally prohibitive, 
users are typically interested only in a relatively small subset 
of the generated molecules that are to be sorted according to some chosen criterion.

\paragraph*{Performance and Scalability}
To assess the runtime of our approach, we conduct experiments with
series of uniformly created molecular formulas of increasing size. 
We consider two patterns:
formulas of the form $\mathit{C}_n\mathit{H}_{2n+2}O$ belong to tree-shaped
molecules (such as ethanol with SMILES \texttt{OCC}), whereas
formulas of the form $\mathit{C}_n\mathit{H}_{2n}O$ require one cycle
(like oxetane, \texttt{C1COC1}) or double bond (like acetone, \texttt{CC(=O)C}).
We use a timeout of 10 minutes for all tools except \textsc{Molgen}, whose free version
is limited to 1 minute runtime. All runs are repeated five times and the median is reported.
Moreover, clingo allows to specify the search strategy via the \texttt{----configuration}
parameter. For the ASP-based approaches, we tested all available options 
on the second pattern (with a 1 minute timeout) 
and found no significant improvements compared to the default \texttt{auto} setting. 

\begin{figure}
\includegraphics[width=\textwidth]{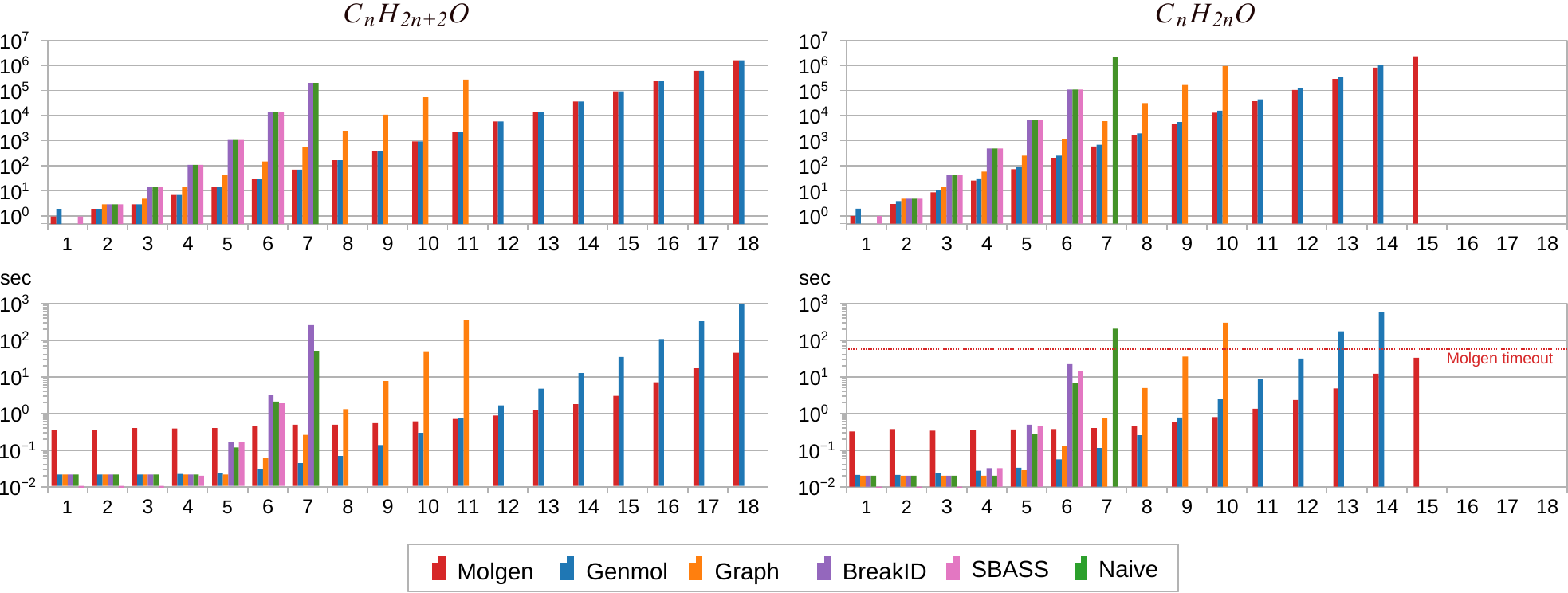}
\caption{Model numbers (top) and runtimes (bottom) for molecules of increasing size}\label{fig_scale}
\end{figure}
The results are shown in Figure~\ref{fig_scale}. As before, \textsc{Molgen} serves as a gold standard.
As seen in the graphs on the top, the number of distinct molecular structures grows exponentially,
and this optimum is closely tracked by \textsc{Genmol} (perfectly for the tree-shaped case).
\textsc{Graph} reduces model counts too, albeit less effectively, whereas \textsc{BreakID} and \textsc{sbass} do not achieve
any improvements over \textsc{Naive} in these structures.

As expected, the runtimes indicate similarly exponential behavior as inputs grow, but 
the point at which computation times exceed the timeout is different for each tool.
\textsc{Molgen} achieves the best scalability overall, whereas \textsc{Genmol} is most scalable
among the ASP-based approaches. \textsc{BreakID} and \textsc{sbass} are even slower than \textsc{Naive}, largely due to longer solving times, whereas the preprocessing of the grounding had no notable impact.

\paragraph*{Learning symmetry-breaking Constraints.}
We also explored an approach proposed by \cite[]{learning}
that aims to automatically learn symmetry-breaking constraints,
which can be appended to the input program as ASP rules.
This method relies on \textsc{sbass} to
derive instance-specific constraints
that are used to 
classify instances as positive or negative examples 
depending on whether they satisfy the derived constraints.
These examples are provided as input to \textsc{ilasp} \cite[]{ILASP_system},
which attempts to learn ASP rules that eliminate negative examples
while preseriving positive ones.

Since \textsc{sbass} did not yield any significant improvement
over the \textsc{Naive} implementation in our case,
we instead tested whether \textsc{ilasp} could 
automatically derive symmetry-breaking constraints
using explicitly constructed examples.
As a representative set of examples,
we used instances of the form 
$\mathit{C}_c\mathit{H}_h\mathit{N}_n\mathit{O}_o$ with $c\in\{3,4\}$, $n,o\in\{0,1\}$, $h=2\cdot c +2+n-2\cdot x$, and $x\in\{0,1,2\}$ cycles.
For each class of isomorphic molecular graphs,
we selected the one with the lexicographically smallest adjacency matrix
as a positive example and randomly sampled a negative example from the others,
yielding a total of 1,228 examples.

The hypothesis space for \textsc{ilasp},
i.e., the set of possible rules,
was configured in such a way 
as to allow the construction of the \textsc{Graph} encoding,
resulting in 388,098 candidate rules.
We found that even with 768 GiB of RAM, 
the process exceeded available memory after approximately nine hours, 
making it infeasible to complete the computation.
A considerably easier version of the problem
where the definition of the \ASPinline{sat} predicate (lines 1--6) 
is given as background knowledge to \textsc{ilasp}
successfully derives the missing constraint (lines 8--9)
in 7.5 minutes.

We conclude that useful symmetry-breaking constraints 
cannot be learned in our setting
without providing a significant amount of background knowledge,
which is unrealistic in practice. 

 \section{Conclusion}

Motivated by the practical problem of interpreting results in mass spectrometry,
we developed an ASP-based approach for enumerating molecular structures
based on partial information about their chemical composition.
We focused on molecular formulas as inputs, but our prototype \textsc{Genmol}
can also use additional signals, such as known molecular fragments.
Indeed, one of the main strengths of an ASP-based solution is that it is
very simple to refine the search space with additional constraints.
Our extensive evaluation showed that our approach improves upon direct ASP-based solutions
and existing optimizations by several orders of magnitude, regarding both performance
and conciseness, bringing it into the range of optimized commercial implementations
that (presumably) lack the flexibility of ASP.

In general, we believe that our conceptual work towards canonical graph representations
and their efficient realization in ASP is relevant beyond our motivating application
scenario. Most of our ideas can be readily generalized to other graph-related search problems,
and may therefore be of interest to ASP practitioners. Moreover, our real-world evaluation data
can be a valuable benchmark for further research on automated symmetry-breaking in ASP.

 \paragraph{Acknowledgements.}
This work is supported by 
Deutsche Forschungsgemeinschaft 
in project number 389792660 
(TRR 248, \href{https://www.perspicuous-computing.science/}{CPEC}),
by the Bundesministerium für Bildung und Forschung 
under European ITEA project 01IS21084 
(\href{https://www.innosale.eu/}{InnoSale}), in the \href{https://www.scads.de/}{Center for Scalable Data Analytics and
Artificial Intelligence} (\href{https://scads.ai/}{ScaDS.AI}),
and by BMBF and DAAD 
in project 57616814 (\href{https://www.secai.org/}{SECAI}).

\paragraph{Disclosure of Interests.}
The authors have no further competing interests to declare. 
\bibliographystyle{tlplike}
\bibliography{references}
\appendix
\section{List of predicates}\label{app:pred-list}

\newcommand{\ASPinlineBox}[1]{\fcolorbox{blue!5}{blue!5}{\ASPinline{#1}}}

\noindent
\begin{tabular}{p{.37\textwidth}|p{.595\textwidth}}
    \textbf{predicate} & \textbf{use} \\
    \hline
    \ASPinlineBox{molecular_formula/2} &
    Associates element symbols with atom counts \\
    \ASPinlineBox{element/3} &
    Element symbol with atomic number and valence \\
    \ASPinline{non_hydrogen_atom_count/1} &
    Node count in hydrogen-supressed molecular graph \\
    \ASPinline{atom/1} &
    Nodes of the graph \\
    \ASPinline{main_chain_len/1} &
    Length of the longest path in the graph \\
    \ASPinline{valence/1} &
    Valence value \\
    \ASPinline{valence_count/2} &
    Number of atoms with given valence value \\
    \ASPinline{element_valence/2} &
    Reducing from muti-valence elements \\
    \ASPinline{element_valence_count/3} &
    Number of atoms of given element type and valence \\
    \ASPinline{num_cycles_or_multi_bonds/1} &
    Degree of unsaturation / Number of additional edges \\
    \ASPinline{num_multi_bonds/1} &
    Total number of additional multi bond edges \\
    \ASPinline{num_cycles/1} &
    Total number of additional cycle edges \\
    \ASPinlineBox{symbol/2} &
    Vertex coloring to associate atoms with element \\
    \ASPinline{element_count/3} &
    Track used-up element symbols (up to atom) \\
    \ASPinline{multi_bond_count/2} &
    Track used-up multi-bonds (up to atom) \\
    \ASPinlineBox{multi_bond/2} &
    Associate multiplicity to edge \\
    \ASPinline{cycle_start_num/2} &
    Track used-up cycle start markers (up to atom) \\
    \ASPinline{cycle_start_count/2} &
    Number of cycle start markers at atom \\
    \ASPinlineBox{cycle_start/2} &
    Cycle start marker \\
    \ASPinline{cycle_end_count/2} &
    Track used-up cycle end markers (up to atom) \\
    \ASPinline{cycle_end_num/2} &
    Number of cycle end markers at atom \\
    \ASPinlineBox{cycle_end/2} &
    Cycle end marker \\
    \ASPinline{preset_bonds/2} &
    Number of binding places blocked due to upwards multi-bond or cycle markers \\
    \ASPinlineBox{branching/2} &
    Number of direct children beneath node \\
    \ASPinline{size/2} &
    Size of subtree beneath node (including) \\
    \ASPinline{depth/2} &
    Depth of subtree beneath node (including) \\
    \ASPinline{postset_bond/2} &
    Number of binding places blocked due to (multi-) bonds to child nodes\\
    \ASPinline{free_bonds/2} &
    Number of remaining binding places, reserved for hydrogen atoms \\
    \ASPinline{parent/3} &
    Parent relation (children ordered by index) \\
    \ASPinline{part_eq/3} &
    Parial equality of $R(G,v)$ tuples up to some position \\
    \ASPinline{eq/2} &
    Subtree equality \\
    \ASPinline{lt/2} &
    Order on subtrees \\
    \ASPinline{max_len/1} &
    Maximum possible chain length in the graph \\
    \ASPinline{chain/2} &
    Transitive chain relation (main-chain only) \\
    \ASPinline{path/3} &
    Transitive relation (side-chain only) with length \\
    \ASPinline{cycle/5} &
    Info on cycles between main-chain and side-chain \\
    \ASPinline{cross_cycle/5} &
    Info on cycles between neighbouring side-chains \\
    \ASPinline{side_chain/2} &
    Depth of side chains \\
\end{tabular}    
 \end{document}